\documentclass[final]{lmcs} %%% last changed 2014-08-20
\usepackage[utf8]{inputenc}
\pdfoutput=1

\usepackage{lastpage}
\lmcsdoi{18}{3}{9}
\lmcsheading{}{\pageref{LastPage}}{}{}%
{Jan.~19,~2021}{Jul.~29,~2022}{}


\usepackage{graphicx} \usepackage{microtype} \usepackage{xspace}
\usepackage{mathtools} \usepackage{stmaryrd} \usepackage{amssymb,amsfonts}
\usepackage{bbm} \usepackage{standalone}
\usepackage{centernot}
\usepackage[ruled,linesnumbered]{algorithm2e}
\usepackage{environ} % Needed to redefine restatable properties in appendix.
\usepackage{thm-restate}
\usepackage{tikz} \usetikzlibrary{arrows, automata, positioning, chains,
  patterns, decorations.pathreplacing, calc}
\usepackage{marvosym}

\usepackage{verbatim}
\usepackage[]{algorithm2e}

\newcommand{\N}{\mathbb{N}} \newcommand{\Z}{\mathbb{Z}} \newcommand{\D}{\mathbb{D}} \newcommand{\Q}{\mathbb{Q}} 
\DeclarePairedDelimiter\size{\lvert}{\rvert}
\DeclarePairedDelimiter\length{\lvert}{\lvert}
\DeclarePairedDelimiter\dist{\lVert}{\rVert}
\newcommand{\set}[1]{\ensuremath{\left\{ #1 \right\}}}
\newcommand{\tuple}[1]{\ensuremath{\left( #1 \right)}}
\DeclarePairedDelimiter\rel{\llbracket}{\rrbracket}
\newcommand{\dom}{\mathrm{dom}}
\newcommand{\prefix}{\leq}
\usepackage{scalerel}
\newcommand{\match}{\mathrel{{\hstretch{.8}{\parallel}}}}
\newcommand{\mismatch}{\ensuremath{\mathrel{\centernot {\match}}}}

 \newcommand{\data}{\mathcal{D}}
\newcommand{\aut}{\mathrm{Aut}}
\newcommand{\test}{\mathsf{Test}}
\newcommand{\keep}{\mathsf{keep}}
\newcommand{\setto}{\mathsf{set}}
\newcommand{\guess}{\mathsf{guess}}
\newcommand{\indata}{\mathsf{input}}
\newcommand{\regop}{\mathsf{update}}
\newcommand{\gguess}{?}
\newcommand{\gsetto}{\downarrow{}}
\newcommand{\gout}{\mathrm{out}~}

\newcommand{\cd}{\mathsf{c}_0}

\newcommand{\cp}{\mathsf{Critical} }

\newcommand{\rn}{\mathcal{R}}

\newcommand{\fo}{\textsf{FO}\xspace}
\newcommand{\qf}{\textsf{QF}\xspace}

\newcommand{\ie}{\textit{i.e.}~}
\newcommand{\eg}{\textit{e.g.}~}
\newcommand{\etc}{\textit{etc}\xspace}

\newcommand{\PSpace}{\textsc{PSpace}\xspace}
\newcommand{\NPSpace}{\textsc{NPSpace}\xspace}
\newcommand{\ExpSpace}{\textsc{ExpSpace}\xspace}
\newcommand{\nra}{\textnormal{\textsf{NRA}}\xspace}

\newcommand{\nrt}{\textnormal{\textsf{NRT}}\xspace}
\newcommand{\ngrt}{\textnormal{\textsf{NGRT}}\xspace}
\newcommand{\drt}{\textnormal{\textsf{DRT}}\xspace}

\usepackage{calc} 
\makeatletter \providecommand*{\shuffle}{%
  \mathbin{\mathpalette\shuffle@{}}%
} \newcommand*{\shuffle@}[2]{%
  \sbox0{$#1\vcenter{}$}%
  \kern .15\ht0 % side bearing
  \rlap{\vrule height .25\ht0 depth 0pt width 2.5\ht0}%
  \raise.1\ht0\hbox to 2.5\ht0{%
    \vrule height 1.75\ht0 depth -.1\ht0 width .17\ht0 %
    \hfill \vrule height 1.75\ht0 depth -.1\ht0 width .17\ht0 %
    \hfill \vrule height 1.75\ht0 depth -.1\ht0 width .17\ht0 %
  }%
  \kern .15\ht0 % side bearing
} \makeatother

 \usepackage[obeyFinal]{todonotes}
\usepackage{etoolbox}
\apptocmd{\sloppy}{\hbadness 10000\relax}{}{}

\begin{document}

\title{Computability of Data-Word Transductions \texorpdfstring{\\}{} over Different Data Domains}
\author[L. Exibard]{L\'eo Exibard\rsuper{a}}
\author[E. Filiot]{Emmanuel Filiot\rsuper{a}}
\author[N. Lhote]{Nathan Lhote\rsuper{b}}
\author[P.-A. Reynier]{Pierre-Alain Reynier\rsuper{b}}

\address{\lsuper{a}Universit\'e Libre de Bruxelles, Brussels, Belgium}
\email{leo.exibard@ulb.ac.be}
\address{\lsuper{b}Aix Marseille Univ, CNRS, LIS, Marseille, France}

\thanks{%
    L.~Exibard was funded by a FRIA fellowship from the F.R.S.-FNRS\@.
    E.~Filiot is a research associate of F.R.S.-FNRS and was supported by the ARC Project Transform F\'ed\'eration Wallonie-Bruxelles and the FNRS CDR J013116F;\@ MIS F451019F projects.
    N.~Lhote and P.-A.~Reynier were partly funded by the ANR projects DeLTA (ANR-16-CE40-0007) and Ticktac (ANR-18-CE40-0015).
}

\begin{abstract}
  In this paper, we investigate the problem of synthesizing computable functions
  of infinite words over an infinite alphabet (data $\omega$-words). The notion
  of computability is defined through Turing machines with infinite inputs which
  can produce the corresponding infinite outputs in the limit. We use
  non-deterministic transducers equipped with registers, an extension of
  register automata with outputs, to describe specifications. Being non-deterministic,
  such transducers may not
  define functions but more generally relations of data $\omega$-words.
  In order to increase the
  expressive power of these machines, we even allow guessing
  of arbitrary data values when updating their registers.

  %For functions over finite alphabets, it is known that the computability notion coincides with   the classical mathematical notion of continuity (for the Cantor distance).
  For functions over data $\omega$-words, we identify a sufficient condition (the possibility of determining the next letter to be outputted, which we call next letter problem) under which
  computability (resp.~uniform computability) and continuity (resp.~uniform continuity) coincide.

  We focus on  two kinds of data domains: first, the general setting of oligomorphic data, which encompasses any data domain with equality, as well as the setting of rational numbers with linear order; and second, the set of natural numbers equipped with linear order. For both settings, we prove that functionality, \emph{i.e.} determining whether the relation recognized by the transducer
  is actually a function, is decidable. We also show that the so-called next letter problem
  is decidable, yielding equivalence between (uniform) continuity and (uniform) computability.
  Last, we provide characterizations of (uniform) continuity,
  which allow us to prove that these notions, and thus
  also (uniform) computability, are decidable. We even show that all these decision problems
  are \PSpace-complete for $(\N,<)$ and for a large class of oligomorphic data domains, including for instance $(\Q,<)$.
\end{abstract}
\maketitle
\clearpage
\tableofcontents
\section*{Introduction}

\subsubsection*{Synthesis} Program synthesis aims at deriving, in an automatic way, a
program that fulfils a given specification. It is very appealing when
for instance the specification describes, in some abstract formalism (an
automaton or ideally a logic), important properties that the program must
satisfy. The synthesised program is then \emph{correct-by-construction} with
regard to those properties. It is particularly important and desirable for the
design of safety-critical systems with hard dependability constraints, which are
notoriously hard to design correctly. In their most general forms,
synthesis problems have two parameters, a set of inputs \textsf{In} and
a set of outputs \textsf{Out}, and relate two classes $\mathcal{S}$ and $\mathcal{I}$ of
specifications and implementations respectively. A specification
$S\in\mathcal{S}$ is a relation $S \subseteq \textsf{In}\times
\textsf{Out}$ and an implementation $I\in\mathcal{I}$ is a function $I
: \textsf{In}\rightarrow \textsf{Out}$. The
$(\mathcal{S},\mathcal{I})$-synthesis problem asks, given a
(finite representation of a) specification $S\in\mathcal{S}$,  whether
there exists $I\in\mathcal{I}$ such that for all $u\in \textsf{In}$,
$(u,I(u))\in S$. If such $I$ exists, then the procedure must return
a program implementing $I$. If all specifications in $\mathcal{S}$ are
\emph{functional}, in the sense that they are the graphs of functions
from $\textsf{In}$ to $\textsf{Out}$, then the
$(\mathcal{S},\mathcal{I})$-synthesis is a membership problem: given
$f\in \mathcal{S}$, does $f\in\mathcal{I}$ hold?

\subsubsection*{Automata-theoretic approach to synthesis} In this paper, we are interested in the automata-theoretic approach to
synthesis, in the sense that specifications and implementations can be
defined by automata, or by automata extended with outputs called
\emph{transducers}. In this approach, \textsf{In} and \textsf{Out} are
sets of words over input and output alphabets $\Sigma$ and $\Gamma$
respectively. Perhaps the most well-known decidable instance of synthesis in
this context is the celebrated result of B\"uchi and
Landweber~\cite{BuLa69}: $\mathcal{S}$ is the class of
$\omega$-regular specifications, which relates infinite input words
$i_0i_1\dots\in \Sigma^\omega$ to infinite output words
$o_0o_1\dots\in\Gamma^\omega$ through $\omega$-automata
(e.g.\ deterministic parity automata), in the sense that the infinite
convolution $i_0o_0i_1o_1\dots \in (\Sigma\Gamma)^\omega$ must be
accepted by an $\omega$-automaton defining the
specification. The class of implementations $\mathcal{I}$ is all the
functions which can be defined by Mealy machines, or equivalently,
deterministic \emph{synchronous} transducers which, whenever they read some input $i\in\Sigma$,
produce some output $o\in\Gamma$ and possibly change their own internal
state. The seminal result of B\"uchi and Landweber has recently triggered a
lot of research in reactive system synthesis and game theory, both on
the theoretical and practical sides, see for
instance~\cite{DBLP:reference/mc/2018}. We identify two important limitations to the now classical setting of $\omega$-regular reactive
synthesis:
\begin{enumerate}
\item[$(i)$] specifications and implementations are required to be synchronous, in
the sense that a single output $o\in \Gamma$ must be produced for each
input $i\in\Sigma$, and
\item[$(ii)$] the alphabets $\Sigma$ and $\Gamma$ are assumed to be finite.
\end{enumerate}
 Let us argue why we believe $(i)$ and $(ii)$ are
indeed limitations. First of all, if a specification is not
realizable by a synchronous transducer, then a classical synthesis algorithm stops with a negative
answer. However, the specification could be realizable in a larger
class of implementations $\mathcal{I}$. As an example, if $S$ is the
set of words $i_0o_0\dots$ such that $o_\ell = i_{\ell+1}$, then $S$
is not realizable synchronously because it is impossible to produce
$o_\ell$ before knowing $i_{\ell+1}$. But this specification is
realizable by a program which can delay its output
production by one time unit. Enlarging the class of implementations can
therefore allow one to give finer answers to the synthesis problem in cases
where the specification is not synchronously realizable. We refer to
this type of relaxations as \emph{asynchronous}
implementations. An asynchronous implementation can be
modelled in automata-theoretic terms as a transducer which, whenever
it reads an input $i\in \Sigma$, produces none or several outputs,
i.e.\ a finite word $u\in \Gamma^*$. Generalizations of reactive
system synthesis to asynchronous implementations have been
considered in~\cite{DBLP:journals/corr/abs-1209-0800,DBLP:conf/csl/FridmanLZ11,DBLP:journals/iandc/WinterZ20}. In
these works however,
the specification is still synchronous, given by an automaton which
strictly alternates between reading input and output
symbols. Here, we also assume that specifications are asynchronous, as it gives more flexibility in the relations that can be expressed. For instance, one can specify that some response has to be delayed, or, when transforming of data streams, allow for erasure and/or duplication of some data values.

The synchronicity
assumption made by classical reactive synthesis is motivated by the
fact that such methods focus on the control of systems rather than on
the data, in the sense that input symbols are Boolean signals issued
by some environment, and output symbols are actions controlling the
system in order to fulfil some correctness properties. From a
data-processing perpective, this is a strong limitation. The
synthesis of systems which need to process streams of data, like a monitoring system or a
system which cleans noisy data coming from sensors, cannot be
addressed using classical $\omega$-regular synthesis. Therefore, one
needs to extend specifications to asynchronous specifications, in the
sense that the specifications must describe properties of executions
which do not strictly alternate between inputs and outputs. Already on
finite words however, the synthesis problem of asynchronous
specifications by asynchronous implementations, both defined by
transducers, is undecidable in general~\cite{CarayolL14}, and
decidable only in some restricted cases~\cite{DBLP:conf/icalp/FiliotJLW16}.
The second limitation $(ii)$ is addressed in the next paragraph.

\subsubsection*{From finite to infinite alphabets} To address the
synthesis of systems where \emph{data} are taken into account, one
also needs to extend synthesis methods to handle infinite
alphabets. As an example, in a system scheduling processes, the data values are process ids. In a stream processing system,
data values can be temperature or pressure measurements for example. Not only
one needs synthesis methods able to handle infinite alphabets of data values,
but where those values can be compared through some predicates, like
equality or a linear order. Recent
works have considered the synthesis of (synchronous) reactive systems processing \emph{data
  words} whose elements can be compared for equality~\cite{DBLP:conf/atva/KhalimovMB18,Ehlers:2014:SI:2961203.2961226,DBLP:conf/concur/KhalimovK19,DBLP:conf/concur/ExibardFR19}
as well as comparison with a linear order on the data~\cite{STACS21}.
To handle data words, just as automata have been extended to \emph{register
  automata}, transducers have been extended to \emph{register transducers}. Such
transducers are equipped with a finite set of registers in which they can store
data values and with which they can compare them for equality, inequality or
in general any predicate, depending on the considered data domain. When
a register transducer reads a data value, it can compare it to the values
stored in its registers, assign it to some register, and output the
content of none or several registers, i.e., a finite word $v$ of register
contents. To have more expressive power, we also allow transducers to
guess an arbitrary data value and assign it to some register. This feature,
called data guessing, is arguably a more robust notion of non-determinism notion for
machines with registers and was introduced to enhance register
automata~\cite{DBLP:journals/ijfcs/KaminskiZ10}. We denote by \nrt the
class of non-deterministic register transducers. As an example,
consider the (partial\footnote{In this paper, data word functions can
  be partial by default and therefore we do not explicitly mention it
  in the sequel.}) data word function $g$
which takes as input any data word of the form $u = su_1su_2\cdots \in
\N^\omega$, $s$ occurs infinitely many times in $u$, and $u_i\in
(\N\setminus\{s\})^+$ for all $i\geq 1$. Now, for all $i\geq
1$, denote by $|u_i|$ the length of $u_i$ and by $d_i$ the last data value
occurring in $u_i$. The function $g$\label{page:functiong} is then defined as $g(u) =
d_1^{|u_1|}sd_2^{|u_2|}s\dots$. This function can be defined by the \nrt
of Figure~\ref{fig:nrtexample}. Note that without the guessing
feature, this function could not be defined by any \nrt.

\begin{figure}[t]
\begin{tikzpicture}[->, >=stealth', auto, node distance=38mm]
  \tikzstyle{every state}=[text=black,fill=yellow!30]

  \node[state, initial, initial text={}] (0) {$i$};
  \node[state, accepting, right= of 0]   (1) {$f$};
  \node[state, right= of 1]   (2) {$o$};

  \path (0) edge node[above] {$\top \mid \gsetto r_s, \gguess r_o, \gout \varepsilon$} (1);
  \path (1) edge node[below] {$\star \neq r_s \mid \gsetto r_c, \gout r_o$} (2);
  \path (2) edge[bend right=100] node[above] {$\wedge \begin{array}{c} \star = r_s \\ r_c = r_o \end{array} \mid \gguess r_o, \gout r_s$} (1);
  \path (2) edge[loop right] node[right] {$\star \neq r_s \mid \gsetto r_c, \gout r_o$} (2);
\end{tikzpicture}
  \caption{An \nrt defining the data word function $g$, equipped with a
    B\"uchi condition. The current data value denoted by $\star$ is tested
    with respect to the content of the registers on the left of the
    bar $|$. On the right of the bar, there are instructions such as
    assigning an arbitrary data value to $r$ (notation $\gguess r$), outputting the
    content of a register or nothing ($\gout r$), or assigning the current data value to
    some register ($\gsetto r$). The B\"uchi condition makes
    sure that the first data value, initially stored in $r_s$ during the
    first transition, occurs infinitely many times. The register $r_c$ stores
    the last data value that has been read. $r_o$
    is meant to store the last data value $d_i$ of an input chunk $u_i$. It has to
    be guessed whenever a new chunk $u_i$ is starting to be read, and on reading
    again $r_s$, the automaton checks that the guess was right by evaluating
    whether $r_c = r_o$ (at that moment, $r_c$ contains
    $d_i$).\label{fig:nrtexample}}
\end{figure}
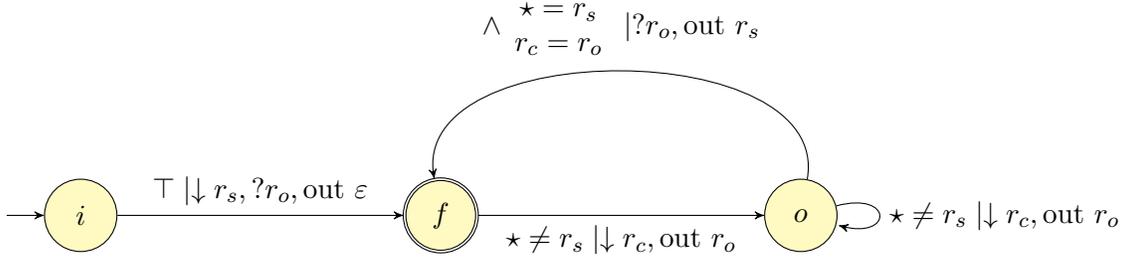

Thanks to the non-determinism of \nrt, in general and unlike the
previous example, there might be several accepting
runs for the same input data word, each of them producing a possibly different output
data word. Thus, \nrt can be used to
define binary relations of data $\omega$-words, and hence
specifications. In the
works~\cite{DBLP:conf/atva/KhalimovMB18,Ehlers:2014:SI:2961203.2961226,DBLP:conf/concur/KhalimovK19,DBLP:conf/concur/ExibardFR19,STACS21}
already mentioned, \nrt have been used as a description of
specifications, however they are assumed to be synchronous and
without guessing.

\subsubsection*{Objective: synthesis of computable data word
  functions} In this paper, our goal is to define a synthesis setting
where both limitations $(i)$ and $(ii)$ are lifted. In particular, specifications are
assumed to be given by (asynchronous) non-deterministic register
transducers equipped with a B\"uchi condition (called \nrt). To retain
decidability, we however make some hypothesis: specifications are
assumed to be functional, i.e., they already define a function from
input data $\omega$-words to output data $\omega$-words. While this
a strong hypothesis, it is motivated by two facts. First, the
synthesis problem of asynchronous implementations from asynchronous specifications given by (non-functional) \nrt
is undecidable in general, already in the finite alphabet
case~\cite{CarayolL14}. Second, functional \nrt
define uncomputable functions in general, and therefore they cannot be
used as machines that compute the function they specify.
Since those functions are
defined over infinite inputs, let us make clear what we mean by
computable functions. A (partial) function $f$ of data $\omega$-words
is computable if there exists a Turing machine $M$ that has an infinite input $x\in \dom(f)$,
and produces longer and longer prefixes of the output $f(x)$ as it reads
longer and longer prefixes of the input $x$. Therefore, such a machine produces
the output $f(x)$ in the limit. As an example, the function
$g$ previously defined is computable. A Turing machine computing it
simply has to wait until it sees the last data value $d_i$ of a chunk $u_i$ (which
necessarily happens after a finite amount of time), compute the length
$\ell_i$ of $u_i$ and once it sees $d_i$, output $d_i^{\ell_i}$ at
once. However, consider the extension $f$ to any input data word defined as follows:
$f(u) = g(u)$ if $u$ is in the domain of $g$, and otherwise $f(u) = s^\omega$ where
$s$ is the first data value of $u$. Such function is not computable. For instance, on input $x = sd^\omega$ (where $d\neq s$ is an arbitrary
data value), we have $f(sd^\omega) =
s^\omega$, as $x$ is not in the domain of $g$.
Yet, on any finite prefix
$\alpha_k = sd^k$ of $sd^\omega$, any hypothetical
machine computing $f$ cannot output anything. Indeed, there exists
a continuation of $\alpha_k$ which is in the domain of $g$, and
for which $f$ produces a word which starts with a different data value than
$f(\alpha_{k}d^\omega)$: it suffices to take the continuation
$(sd)^\omega$, as we have $f(\alpha_k(sd)^\omega) =
g(\alpha_k(sd)^\omega) = d^k(sd)^\omega$.

In this paper, our goal is therefore to study the following synthesis problem: given a
functional \nrt defining a function $f$ of data $\omega$-words, generate a Turing machine
which computes $f$ if one exists. In other words, one wants to decide
whether $f$ is computable, and if it is, to synthesize an algorithm
which computes it.

% In the previously mentioned works, both
%for finite or infinite alphabets, implementations are considered to be
%deterministic transducers. Such an implementation is guaranteed to use only a
%constant amount of memory (assuming data have size $O(1)$). While it makes sense%
%with regards to memory-efficiency, some problems turn out to be undecidable, as
%already mentioned: realisability of $\nrt_\syn$ specifications by $\drt_\syn$,
%or, in the finite alphabet setting, when both the specification and
%implementation are asynchronous. I

\subsubsection*{Contributions} Register transducers can be
parameterized by the set of data values from which the $\omega$-data words are built,
along with the set of predicates which can be used to test those
values. We distinguish a large class of data sets for which we obtain
decidability results for the later problem, namely the class of
oligomorphic data sets~\cite{DBLP:journals/corr/BojanczykKL14}. Briefly, oligomorphic data sets are
countable sets $D$ equipped with a finite set of predicates which satisfies
that for all $n$, $D^n$ can be partitioned into finitely many equivalence classes
by identifying tuples which are equal up to automorphisms
(predicate-preserving bijections). For example, any set equipped with
equality is oligomorphic, such as $(\mathbb{N}, \{=\})$,
$(\mathbb{Q}, \{<\})$ is oligomorphic
while $(\mathbb{N}, \{<\})$ is not.
However $(\mathbb{N}, \{<\})$ is an
interesting data set in and of itself. We also investigate \nrt
over such data set, using the fact that it is a substructure of $(\Q, \{<\})$ which is oligormorphic.
 %Since $(\mathbb{N}, \{<\})$ is an important data set in the sense that it canonically represents countable sets with a non-dense linear order, we also investigate \nrtover such data set.
Our detailed contributions are the following:
\begin{enumerate}

\item We first establish a general correspondence between computability and
  the classical mathematical notion of continuity (for the Cantor
  distance) for functions of data $\omega$-words  (Theorems~\ref{thm:comp2cont}
  and~\ref{thm:cont2comp}). This correspondence holds under
  a general assumption, namely the decidability of what we called the
  \emph{next-letter problem}, which in short asks that the next data value which
  can be safely outputted knowing only a finite prefix of the input
  data $\omega$-word is computable, if it exists. We also show similar correspondences for
  more constrained computability and continuity notions, namely Cauchy,
  uniform and $m$-uniform computability and continuity. In these
  correspondences, the construction of a Turing machine computing the
  function is effective.
\item We consider a general computability assumption for oligomorphic
  data sets, namely that they have decidable first-order
  satisfiability problem~\cite{Bojanczyk19}. We call such data sets \emph{decidable}. We
  then show that functions defined by \nrt over decidable oligomorphic data
  sets  and over $(\mathbb{N},
  \{<\})$, have decidable next-letter problem.   As a consequence
  (Theorems~\ref{thm:compcont} and~\ref{thm:compcontN}), we
  obtain that a function of data $\omega$-words definable by an \nrt over decidable oligomorphic data sets  and over $(\mathbb{N},\set{<})$, is
  computable iff it is continuous (and likewise for all computability
  and continuity notions we introduce). This is a useful mathematical
  characterization of computability, which we use to obtain our main result.

\item As explained before, an \nrt may not define a function in general but
  a relation, due to non-determinism. Functionality is a semantical, and not
  syntactical, notion. We nevertheless show that checking whether an \nrt defines a function is
  decidable for decidable oligomorphic data sets (Theorem~\ref{thm:decoligo}). This problem is called the \emph{functionality problem}
  and is a prerequisite to our study of computability, as we assume
  specifications to be functional.  We establish \textsc{PSpace}-completeness of
  the functionality problem for \nrt over $(\mathbb{N},\{<\})$
  (Corollary~\ref{thm:funN}) and for oligomorphic
  data sets (Theorem~\ref{thm:decoligo}) under some additional assumptions on the data set that we call polynomial
  decidability. In short, it is required that the data set has
  \textsc{PSpace}-decidable first-order satisfiability problem.

\item Finally, we show (again Theorem~\ref{thm:decoligo}) that  continuity of functions defined by
  \nrt over decidable (resp.\ polynomially decidable) oligomorphic data
  sets is decidable (resp.\ \textsc{PSpace}-c). We also obtain
  \textsc{PSpace}-completeness in the non-oligomorphic case
  $(\mathbb{N}, \{<\})$ (Theorem~\ref{thm:contN}). These results also hold for the stronger notion of uniform continuity (see also Theorem~\ref{thm:decUnifContN}). As a result of the correspondence between
  computability and continuity, we also obtain that computability and uniform computability are decidable
  for functions defined by \nrt over decidable oligomorphic data sets,
  and \textsc{PSPace}-c for polynomially decidable oligomorphic data
  sets as well as $(\mathbb{N}, \{<\})$. This is our main result
  and it answers positively our initial synthesis motivation.
\end{enumerate}

\noindent
The proof techniques we use have the following structure in common:
first, we characterize non-functionality and non-continuity by
structural patterns on \nrt and establish small witness properties for
the existence these patterns. Then, based on the small witness
properties, we show how to decide whether given an \nrt, such patterns
are matched or not. While the proofs have some similarities between
the oligomorphic case, the case $(\mathbb{N},\{<\})$ and the functionality
and continuity problems, there are subtle
technical differences which make them hard to factorize with
reasonable amount of additional notations and theoretical
assumptions.

\subsubsection*{Related Work} We have already mentioned works related
to the synthesis problem. We now give references to results on
computability and continuity. The notion of continuity with regards to Cantor
distance is not new, and for rational functions over finite alphabets, it was
already known to be decidable~\cite{DBLP:journals/tcs/Prieur02}.
The approach of Prieur is to reduce continuity to functionality by defining from a transducer $T$ a transducer realizing its topological closure. We were able to extend this approach to almost all the cases we considered, except for transducers over $(\N,\set{<})$ with guessing allowed, so we chose a different proof strategy.
The connection between continuity and computability for functions of $\omega$-words over a finite alphabet has
recently been investigated in~\cite{DBLP:conf/concur/DaveFKL20} for
one-way and two-way transducers. Our results lift the case of one-way
transducers from~\cite{DBLP:conf/concur/DaveFKL20} to data words. Our
results were partially published in conference
proceedings~\cite{DBLP:conf/fossacs/ExibardFR20}. In that
publication, only the case of data sets equipped with the equality
predicate was considered. We now consider oligomorphic data sets
(which generalise the latter case), the data set $(\mathbb{N},\{<\})$ and
new computability notions. Despite the fact that our
results are more general, this generalisation also allows to extract
the essential arguments needed to prove this kind of
results. Moreover, compared to~\cite{DBLP:conf/fossacs/ExibardFR20},
we add here the possibility for the register transducer to make
non-deterministic register assignment (data guessing), which strictly
increases their expressive power.

\section{Data alphabet, languages and transducers}
% \nathan{Notations: I try to use $x,y,z$ for $\omega$-words and $u,v,w$ for finite or $\infty$-words.}

\subsection{Data as logical structures}
Let $\Sigma$ be a finite signature with relation and constant
symbols. Let $\D=(D,\Sigma^\D)$ be a logical structure over $\Sigma$
with a countably infinite domain $D$ and an interpretation of each symbol of $\Sigma$. Note that we often identify $\D$ and $D$ when the structure considered is clear, from context.

An \emph{automorphism} of a structure $\D$ is a bijection $\mu: D\rightarrow D$ which preserves the constants and the predicates of $\D$: for any constant $c$ in $D$, $\mu(c)=c$ and for any relation of $\Sigma$, $R\subseteq D^r$, we have $\forall\ \bar x,\  R(\bar x) \Leftrightarrow R(\mu(\bar x))$, where $\mu$ is naturally extended to $D^r$ by applying it pointwise. We denote by $\aut(\D)$ the set of automorphisms of $\D$.
Let $\bar x\in \D^d$, the set $\set{\mu(\bar x)\mid\ \mu\in \aut(\D)}$ is called the \emph{orbit} of $\bar x$ under the action of $\aut(\D)$.

We will be interested in structures that have a lot of symmetry. For instance the structures $(\N,\set{0,=})$, $(\Z,\set{<})$ and $(\Q,\set{<})$ fall under our study as well as more sophisticated structures like $(1(0+1)^*,{\otimes})$ where $\otimes$ is the bitwise xor operation. Other structures like $(\Z,\set{+})$ will not have enough internal symmetry to be captured by our results.

\begin{defi}
  A logical structure $\D$ is \emph{oligomorphic} if for any natural number $n$ the set $\D^n$ has finitely many orbits under the action of $\aut(\D)$.
\end{defi}

\begin{exa}%
  \label{exa:oligomorphic}
Oligomorphic structures can be thought of as ``almost finite''. Consider $(\N, \set{=})$, then $\N^2$ only has two orbits: the diagonal $\set{(x,x)\mid\ x\in \N}$ and its complement $\set{(x,y)\in \N^2\mid\ x\neq y}$. In fact $(\N, \set{=})$ is oligomorphic, since the orbit of an element of $\N^n$ is entirely determined by which coordinates are equal to each other. Similarly, one can see that the dense linear order $(\Q,\set{<})$ is oligomorphic.

The automorphism group of $(\Z,\set{<})$ consists of all translations $n \mapsto n+c$ for some fixed $c \in \Z$. This means that $\Z$ only has one orbit. However, $\Z^2$ has an infinite number of orbits since the difference between two numbers is preserved by translation. Hence $(\Z,\set{<})$ is not oligormorphic. However, the fact that $(\Z,\set{<})$ is a substructure of $(\Q,\set{<})$ will allow us to extend our results to this structure, with some additional work.
For more details on oligomorphic structures see~\cite[Chap.~3]{Bojanczyk19}.
\end{exa}

Let $G$ be a group acting on both $X,Y$, then a function $f:X\rightarrow Y$ is
called \emph{equivariant} if for all $x\in X, \mu\in G$ we have $f(\mu(x))=
\mu(f(x))$.
%A function $f:\D^\omega\rightarrow \D^\omega$ is called \emph{equivariant} if for all $x\in A^\omega, \mu\in \aut(\D)$ we have $f(\mu(x))= \mu(f(x))$.
%Given a set $S\subseteq \D$, an \emph{$S$-automorphism} is an element $\mu\in \aut(\D)$ such that for all $s\in S$ we have $\mu(s)=s$; we denote by $S$-$\aut(\D)$ the subgroup of $S$-automorphisms. The function $f$ is called \emph{$S$-supported} if for all $x\in A^\omega, \mu\in S$-$\aut(\D)$ we have $f(\mu(x))= \mu(f(x))$. An equivariant function is $\varnothing$-supported, any function is $\D$-supported and any function which is supported by a finite set is called \emph{finitely supported}.

\subsection{Words and data words}
For a (possibly infinite) set $A$, we denote by $A^*$ (resp. $A^\omega$) the set
of finite (resp.\ infinite) words over this alphabet, and we let $A^\infty = A^*
\cup A^\omega$. For a word $u = u_1 \dots u_n$, we denote $\length{u} = n$ its length,
and, by convention, for $x \in A^\omega, \length{x} = \infty$. The empty word is denoted $\varepsilon$. For $1\leq i\leq
j \leq \length{w}$, we let $w[i{:}j] = w_i w_{i+1} \dots w_j$ and $w[i] =
w[i{:}i]$ the $i$th letter of $u$. For $u,v \in A^\infty$, we say that $u$ is a
prefix of $v$, written $u \prefix v$, if there exists $w \in A^\infty$ such that
$v = uw$. In this case, we define $u^{-1}v = w$. For $u,v \in A^\infty$, we say
that $u$ and $v$ \emph{match} if either $u\prefix v$ or $v \prefix u$, which we denote by $u\match v$, and we say that they \emph{mismatch}, written $u\mismatch v$, otherwise. Finally, for $u,v \in A^\infty$, we denote by $u \wedge v$
their longest common prefix, i.e.\ the longest word $w \in A^\infty$ such that $w
\prefix u$ and $w \prefix v$.

Let $\D$ be a logical structure. A word over
$\D$ is called a \emph{$\D$-data word} (or just \emph{data
  word}). Note that $\aut(\D)$ naturally acts on $\D^\infty$.
\subsection{Functions and relations}
A (binary) relation between sets $X$ and $Y$ is a subset $R \subseteq X \times Y$. We denote its domain $\dom(R) = \{x \in X \mid \exists y \in Y, (x,y) \in R\}$. It is \emph{functional} if for all $x \in \dom(R)$, there exists exactly one $y \in Y$ such that $(x,y) \in R$. Then, we can also represent it as the function $f_R : \dom(R) \rightarrow Y$ such that for all $x \in \dom(R)$, $f(x) = y$ such that $y \in Y$ (we know that such $y$ is unique). $f_R$ can also be seen as a \emph{partial} function $f_R : X \rightarrow Y$.
\begin{conv}
In this paper, unless otherwise stated, functions of data words are assumed to be partial, and we denote by $\dom(f)$ the domain of any (partial) function $f$.
\end{conv}

\subsection{Register transducers}
Let $\D$ be a logical structure, and let $R$ be a finite set of variables.
%We consider the structure $\D_\bot$ which is simply $\D$ with one extra constant $\bot$ which does not belong to any relation of $\D$ except of course $\bot=\bot$. This extra symbol is supposed to represent empty registers and is needed for annoying technical reasons: namely when we want to initialize regiters deterministically but $\D$ does not have any constant.
We define \emph{$R$-tests} by the following grammar:
\[\phi::= P(\bar t)|\phi\wedge \phi| \phi\vee \phi|\neg \phi\]
where $P$ is a symbol of arity $k$ in the signature of $\D$ and $\bar t$ a $k$-tuple of terms. We denote by $\test(R)$ the set of $R$-tests.
Terms are defined by either a constant of $\D$ or a variable of $R$.
In other words $R$-tests are exactly the quantifier-free formulas over the signature of $\D$ using variables in $R$.
\begin{rem}
We choose tests to be quantifier-free formulas. However we could have chosen existential first-order formulas without affecting our results. Note that we choose this formalism just for simplicity's sake, and that it does not make any difference for structures which admit quantifier elimination such as $(\N,\set{=})$ or $(\Q,\set{<})$.
\end{rem}

A \emph{non-deterministic register transducer} (\nrt for short) over $\D$ is a tuple $\tuple{Q,R,\Delta,q_0,\overline{\cd},F}$. Where $Q$ is a finite set of states, $R$ is a finite set of registers, $q_0\in Q$, $\overline{\cd}\in \Sigma^R$ is a vector of constant symbols, $F\subseteq Q$, and $\Delta$ is a finite subset of

\[\underbrace{Q}_{\text{current state}}\times \underbrace{\test(R\uplus \set{\indata})}_{\text{current registers + input data}}\times \underbrace Q_{\text{target state}}\times \underbrace {\set{\keep,\setto, \guess}^R}_{\text{register operations}}\times \underbrace{R^*}_{\text{output word}}\]
A \emph{non-guessing} transducer (\ngrt) has a
transition function which is included in $Q\times \test(R\uplus \set{\indata})\times Q\times \set{\keep,\setto}^R\times R^*$.
Finally, a \emph{deterministic} transducer (\drt) satisfies an even stronger condition:
its transition relation is a function of type $\delta: Q\times \test(R\uplus \set{\indata})\rightarrow Q\times \set{\keep,\setto}^R\times R^*$ where, additionally, tests are mutually exclusive, i.e.\ for all $\phi,\psi \in \dom(\delta)$, $\phi \wedge \psi$ is unsatisfiable.

\begin{rem}
 Note that in the definition of a transducer we require that $\D$ contains at least one constant symbol.
 This is needed for annoying technical reasons, namely in order to initialize registers to some value.

 However it is not too damaging since,
 given a $\Sigma$-structure $\D$ of domain $D$, one can always consider the $\Sigma \uplus \set{\cd}$-structure $\D_\bot$ with domain $D\uplus \set{\bot}$, which is just the structure $\D$ with the extra constant symbol being interpreted as the new element $\bot$, the other relations and constants are unchanged, except naturally for the equality relation which is extended to include $(\bot,\bot)$.

For simplicity's sake we will sometimes talk about structures without mentioning any constant, implicitely stating that we extend the structure to include $\bot$.
Also note that this operation of adding a fresh constant does not affect oligomorphicity.
\end{rem}

Let $T$ be an \nrt given as above. A \emph{configuration} $C$ of $T$ is given by a pair $(q,\bar d)$ where $q\in Q$ is a state and $\bar d\in \D^R$ is a tuple of data values, hence the group $\aut(\D)$ naturally acts on the configurations of $T$ by not touching the states and acting on the content of the registers pointwise. The \emph{initial configuration} is the pair $C_0=(q_0,\bar d_0)$ with $\bar d_0=\overline \cd^\D$ being the interpretation of the constants in $\D$. A configuration is called \emph{final} if the state component is in $F$.
Let $C_1=(q_1,\bar d_1),C_2=(q_2,\bar d_2)$ be two configurations, let $d\in \D$ and let $t=(q_1,\phi, q_2,\regop,v)\in \Delta$.
We say that $C_2$ is a \emph{successor configuration} of $C_1$ by reading $d$ through $t$ and producing $w\in \D^{|v|}$ if the following hold:
\begin{itemize}
\item By letting $\bar d_1[\indata \leftarrow d]: \left\{ \begin{array}{l} r \in R \mapsto \bar d_1(r) \\ \indata \mapsto d \end{array} \right.$, we have $\bar{d}_1[\indata \leftarrow d] \models \phi$
\item for all $r\in R$, if $\regop(r)=\keep$, then $\bar{d}_2(r)=\bar d_1(r)$
\item for all $r\in R$, if $\regop(r)=\setto$, then $\bar{d}_2(r)=d$
\item $w(i)=\bar d_2(v(i))$ for all $i\in \set{1,\ldots, |v|}$ % chktex 36
\end{itemize}
Moreover, we write $C_1\xrightarrow{d,\phi,\regop | w }_T C_2$ to indicate that fact. Often we don't mention $T$ (when clear from context), nor $\phi$ and $\regop$, and we simply write $C_1\xrightarrow{d| w } C_2$. Given a sequence of successor configurations, called a \emph{run}, $\rho=C_1\xrightarrow{d_1| w_1 } C_2 \xrightarrow{d_2| w_2 } C_3\ldots C_n \xrightarrow{d_n| w_n } C_{n+1}$, we write $C_1\xrightarrow{d_1d_2\cdots d_n| w_1w_2\cdots w_n } C_{n+1}$. We sometimes even don't write the output $C\xrightarrow {u} C'$ stating that there is a sequence of transitions reading $u$ going from $C$ to $C'$. % chktex 11

Let $\rho=C_1\xrightarrow{d_1|v_1} C_2 \ldots C_n \xrightarrow{d_n|v_n} C_{n+1}\ldots$ denote a possibly infinite run. If $C_1=C_0$, then $\rho$ is called \emph{initial}. If an infinite number of configurations of $\rho$ are final, we say that $\rho$ is \emph{final}. A run which is both initial and final is \emph{accepting}. We say that the run $\rho$ is over the input word $x=d_1\ldots d_n\ldots$ and produces $w=v_1\ldots v_n\ldots$ in the output. Then the semantics of $T$ is defined as $ \rel T =\set{(x,w)\mid\ \rho \text{ is over $x$, produces $y$ and is accepting}}\subseteq \D^\omega\times \D^\infty$. An \nrt is called \emph{functional} if $\rel T$ is a (partial) function. Note that in the following we will mainly consider transducers that only produce $\omega$-words. Restricting the accepting runs of a transducer to runs producing infinite outputs is a B\"uchi condition and can easily be done by adding one bit of information to states.

%\leo{Why is it now $v$? (was $w$ before).}

%\subsection{Closure under composition}
\section{Continuity and computability}

\subsection{Continuity notions}

We equip the set $A^\infty$ with the usual distance: for
$u,v \in A^\infty$, $\dist{u,v} = 0$ if $u=v$ and $\dist{u,v} =
2^{-\length{u \wedge v}}$ otherwise. A sequence of (finite or infinite) words
$(w_n)_{n\in\mathbb{N}}$ converges to some word $w$ if for all
$\epsilon>0$, there exists $N\ge0$ such that for all $n\ge N$, $\dist{w_n,w}\le
\epsilon$. Given a language $L\subseteq A^\infty$, we denote by $\bar L$ its topological closure, \ie the set of words which can be approached arbitrarily close by words of $L$.

\begin{rem}%
  \label{rem:infiniteNotCompact}
Whether the alphabet $A$ is finite or infinite substantially modifies the properties of the metric space $A^\infty$. Indeed when $A$ is finite this space is compact, but it is not when $A$ is infinite.
\end{rem}

\subsubsection{Continuity}

\begin{defi}[Continuity]
  A function $f:A^\omega \rightarrow B^\omega$ is \emph{continuous} at $x \in \dom(f)$ if (equivalently):
  \begin{enumerate}[(a)]
  \item for all sequences of words $(x_n)_{n \in \N}$ converging towards
    $x$, where for all $i \in \N$, $x_i \in \dom(f)$, we have that $(f(x_n))_{n
      \in \N}$ converges towards $f(x)$.
  \item  $\forall i \geq 0,\ \exists j,\ \forall y \in \dom(f),\ \length{x\wedge y
      } \geq j \Rightarrow \length{f(x) \wedge f(y)} \geq i$
  \end{enumerate}
  The function $f$ is called \emph{continuous} if it is continuous at each $x \in
  \dom(f)$.
  %Finally, a functional \nrt $T$ is \emph{continuous} when $\rel{T}$  is continuous.
\end{defi}

\subsubsection{Cauchy continuity}

A Cauchy continuous function maps any Cauchy sequence to a Cauchy
sequence. One interesting property of Cauchy continuous functions is
that they always admit a (unique) continuous extension to the
completion of their domain. Since we deal with $A^\infty$ which is
complete, the completion of the domain of a function $f$, denoted
$\overline{\dom(f)}$, is simply its
closure.

\begin{defi}[Cauchy continuity]%
  \label{def:cauchy-cont}
  A function $f:A^\omega \rightarrow A^\omega$ is \emph{Cauchy continuous} if the image of a Cauchy sequence in $\dom(f)$ is a Cauchy sequence.
  %Given a Cauchy continuous function $f$, we denote by $\bar f$ its continuous extension.
  %Finally, a functional \nrt $T$ is \emph{Cauchy continuous} when $\rel{T}$  is Cauchy continuous.
\end{defi}

\begin{rem}%
  \label{rem:cauchy-cont}
  Any Cauchy continuous function $f$ can be continuously extended over
  $\overline{\dom(f)}$ in a unique way, which we denote by $\bar f$.
\end{rem}

\subsubsection{Uniform continuity}

\begin{defi}[Uniform continuity]
  A function $f:A^\omega \rightarrow A^\omega$ is \emph{uniformly continuous} if
 there exists a mapping $m:\N\rightarrow \N$ such that:  \[\forall i \geq 0, \forall x,y \in \dom(f), \length{x
  \wedge y} \geq m(i) \Rightarrow \length{f(x) \wedge f(y)} \geq i\]
  Such a function $m$ is called a \emph{modulus of continuity}\footnote{The usual notion of modulus of continuity is defined with respect to distance, but here we choose to define it with respect to longest common prefixes, for convenience. Given $m$ a modulus of continuity in our setting we can define $\omega:x\mapsto 2^{-m\left(\left\lceil\log_2\left(\frac 1 x\right)\right\rceil\right)}$ and recover the usual notion.} for $f$. We also say that $f$ is \emph{$m$-continuous}.
  Finally, a functional \nrt $T$ is \emph{uniformly continuous} when $\rel{T}$
  is uniformly continuous.
\end{defi}

\begin{rem}\label{rem:unifcont}
In the case of a finite alphabet, and in general for compact spaces, Cauchy continuity is equivalent to uniform continuity, but for infinite alphabets this does not hold anymore. Consider the following function $f$ computable by a \drt over the data alphabet $(\N,\set{0,<})$ and defined by $u0x\mapsto x$, for $u0\in \N^*$ being a strictly decreasing sequence. Then this function is not uniformly continuous, since two words may be arbitrarily close yet have very different images. However one can check that the image of a Cauchy sequence is indeed Cauchy: let $(x_n)_{n \in \N}$ be a Cauchy sequence in the domain of $f$. Let us assume without loss of generality that all the $x_n$'s begin with the same letter $i\in\N$. Then, after reading at most $i+1$ symbols of one of the $x_n$'s, the \drt outputs something. Let $j\in \N$ and let $N$ be such that for all $m,n\geq N$ we have $|x_m\wedge x_n|\geq i+j+1$. Thus we have $|f(x_m)\wedge f(x_n)|\geq j$, which means that $(f(x_n))_{n \in \N}$ is Cauchy.
\end{rem}

We've seen in the previous remark that Cauchy continuity and uniform continuity don't coincide over infinite alphabets. However when dealing with oligomorphic structures we recover some form of compactness, that is compactness of $\D^\infty/{\aut(\D)}$, which ensures that the two notions do coincide in this case.%\footnote{This holds for equivariant functions, and indeed for finitely supported functions since the space $A^\infty/S$-$\aut(\D)$ is also compact for $S$ finite.}
\begin{prop}
  Let $\D$ be an oligomorphic structure and let $f:\D^\omega\rightarrow \D^\omega$ be an equivariant function. Then $f$ is uniformly continuous if and only if it is Cauchy continuous.

\end{prop}
\begin{proof}
  It is clear that a uniformly continuous function is in particular Cauchy continuous.
  Let $\D$ be an oligomorphic structure and let $f:\D^\omega\rightarrow \D^\omega$ be an equivariant function. %We show the result in the equviariant case, but the same approch work for finitely supported functions.
  Let us assume that $f$ is not uniformly continuous. This means that there exists $i\in \N$ and a sequence $(x_n,y_n)_{n\in \N}$ such that for all $n$, $|x_n\wedge y_n|\geq n$ and $|f(x_n)\wedge f(y_n)|\leq i$.
  Let us consider the sequence $([x_n],[y_n])_{n\in \N}$ of pairs of elements in $\D^\omega/{\aut(\D)}$, \ie words are seen up to automorphism. Using standard arguments, one can show that the set $\D^\omega/{\aut(\D)}$, equipped with the distance $d([x],[y]) = \min_{u \in [x],v \in [y]} d(u,v)$, is compact (see~\cite[Proposition 12.28]{exibard:tel-03409602} for details). As a consequence, we can extract a subsequence (which we also call $([x_n],[y_n])_{n\in \N}$ for convenience) and which is convergent.
  This means that there are automorphisms $(\mu_n)_{n\in \N}$ such that the sequence
  $(\mu_n(x_n))_{n\in \N}$  (and thus $(\mu_n(y_n))_{n\in \N}$) converges. Hence
  by interleaving $(\mu_n(x_n))_{n\in \N}$ and $(\mu_n(y_n))_{n\in \N}$, we
  obtain a converging sequence whose image is divergent, which means that $f$ is
  not Cauchy continuous. \qedhere

\end{proof}

\begin{rem}
In order to refine uniform continuity one can study $m$-continuity for particular kinds of functions $m$.
For instance for $m:i\mapsto i+b$, $m$-continuous functions are exactly $2^b$-Lipschitz continuous functions. Similarly, for $m: i\mapsto ai+b$, $m$-continuous function are exactly the $\frac 1 a$-H\"older continuous functions.

Note that while these notions are interesting in and of themselves, they are very sensitive to the metric that is being used. For instance the metric $d(x,y)=\frac 1 {|x\wedge y|}$ while defining the same topology over words, yields different notions of Lipschitz and H\"older continuity.
\end{rem}

\subsection{Computability notions}

Let $\D$ be a data set. In order to reason with computability, we assume in the sequel that the countable
set of data values $\D$ we are dealing with has an effective
representation, meaning that each element can be represented in a
finite way.
For instance, this is the case when $\D = \mathbb{N}$. Moreover, we assume that checking if a tuple of values belongs to some relation of $\D$ is decidable. We say that the structure $\D$ is \emph{representable}.
Formally, a structure is representable if there exists a finite
alphabet $A$ and an injective function $\mathsf{enc} : \D\rightarrow
A^*$ such that the sets $\{\mathsf{enc}(d)\mid d\in \D\}$, $\{\mathsf{enc}(c)\mid c\text{ is a constant of
}\D\}$ and $\set{\mathsf{enc}(d_1)\sharp \cdots \sharp
  \mathsf{enc}(d_k)\mid\ (d_1,\ldots, d_k)\in R}$ are decidable for all predicates
$R$ of $\D$ and $\sharp\not\in A$. Any infinite word
$d_1d_2\dots\in \D^\omega$ can be encoded as the $\omega$-word
$\mathsf{enc}(d_1)\sharp\mathsf{enc}(d_2)\sharp\dots\in
(A^*\sharp)^\omega$.

We now define how a Turing machine can compute
a function from $\D^\omega$ to $\D^\omega$. We consider deterministic
Turing machines whose cells can contain a letter from $A\cup \{\sharp\}$ or a letter
from a finite working alphabet. They have three tapes: a read-only
one-way input tape on alphabet $A\cup \{\sharp\}$ (containing an encoding of an infinite input data word), a two-way working
tape, and a write-only one-way output tape on alphabet $A\cup
\{\sharp\}$ (on which they write the
encoding of the infinite output data word). Since we always work
modulo encoding, for the sake of simplicity, from now on and in the
rest of the paper, we assume that each cell of the Turing machine, on
the input and output tapes, contain a data value $d\in\D$, while cells of
the working tape are assumed to contain either a data value $d\in \D$ or a
letter from the working alphabet. So, instead of saying the input contains
the encoding of a data word $x$, we just say that it contains the
input data word $x$. We discuss in
Remark~\ref{rem:encodings} how the computability
notions we introduce hereafter are sensitive to encodings.

Consider such a Turing machine $M$ and some input data word $x\in \D^\omega$. For any integer $k\in\mathbb{N}$, we let $M(x,k)$ denote the
finite output data word written by $M$ on its output tape after reaching cell number $k$
of the input tape (assuming it does).
Observe that as the output tape is write-only, the sequence
of data words $(M(x,k))_{k\ge 0}$ is non-decreasing, and thus we
denote by $M(x)$ the limit content of the output tape. %Moreover, for convenience's sake, we set $M(x,k)$ to $M(x)$ if cell number $k$ is never reached in the computation.

\begin{defi}[Computability]
Let $\D$ be a representable data domain.
  A data word function $f:\D^\omega \rightarrow \D^\omega$ is \emph{computable} if there exists a deterministic multi-tape
  machine $M$ such that for all $x\in\dom(f)$, $M(x)=f(x)$. We say that $M$ \emph{computes} $f$.
\end{defi}

\begin{defi}[Cauchy computability]
Let $\D$ be a representable data domain.
  A data word function $f:\D^\omega \rightarrow \D^\omega$ is \emph{Cauchy computable} if there exists a deterministic multi-tape
  machine $M$ computing $f$ such that for all $x$ in the topological
  closure $\overline{\dom(f)}$ of $\dom(f)$, the sequence $(M(x,k))_{k\ge 0}$
  converges to an infinite word.
  In other words a Cauchy computable function is a function which admits a continuous extension to the closure of its domain and which is computable. We say that $M$ \emph{Cauchy computes} $f$.
\end{defi}

\begin{defi}[Uniform computability]
Let $\D$ be a representable data domain.
  A data word function $f:\D^\omega \rightarrow \D^\omega$ is \emph{uniformly computable} if there exists a deterministic multi-tape
  machine $M$ and a computable mapping $m:\N\rightarrow \N$ such that
  $M$ computes $f$ and for all $i\geq 0$ and $x\in\dom(f)$, $|M(x,m(i))|\geq i$. Such a function $m$ is called a \emph{modulus of computability} for $f$. In that case $f$ is called \emph{$m$-computable}. We say that $M$ \emph{uniformly computes} $f$, and also that $M$ \emph{$m$-computes} $f$.
\end{defi}

\begin{exa}
  The function $g$ defined in the Introduction (p.~\pageref{page:functiong}), for the data domain of
  integers, is computable. Remind that it is defined on all input
  words of the form $x = su_1d_1su_2d_2s\dots$ such that $s$ occurs
  infinitely often, and for all $i$, $s$ does not occur in $u_{i}d_i$,
  by $g(x) = d_1^{|u_1|+1}sd_2^{|u_2|+1}s\dots$. A Turing machine just
  needs to read the input up to $d_1$, then output $d_1$ exactly
  $|u_1|+1$ times, and so on for the other pieces of inputs.

  In Remark~\ref{rem:unifcont}, the function $f$ is not uniformly
  continuous but Cauchy continuous. It is actually not uniformly
  computable but Cauchy computable. As a matter of fact, all the
  computability notions we define here entail
  the respective continuity notions defined before. We make this
  formal in Section~\ref{sec:compvscont}.
\end{exa}

\begin{rem}[Robustness to encoding]\label{rem:encodings}
  When actually representing words over an infinite alphabet, it is
  not realistic to assume that one letter takes a constant amount of
  space and can be read in a constant amount of time. Then, which of
  the many notions introduced above are sensitive to encoding and
  which ones are more robust?

  Let $\D$ be a representable structure, and let
  $\mathsf{enc}:\D\rightarrow A^*$ be its encoding function. Let $f:\D^\omega
  \rightarrow \D^\omega$ be a function and let $f_{\mathsf{enc}}:
  (A\uplus\sharp)^\omega \rightarrow (A^*\uplus\sharp)^\omega$ be defined
  as ${\mathsf{enc}_\sharp}\circ f\circ {\mathsf{enc}_\sharp}^{-1}$, where
  $\mathsf{enc}_\sharp:d_1\cdots d_n \mapsto \mathsf{enc}(d_1)\sharp \cdots\sharp \mathsf{enc}(d_n)$. Continuity and
  computability are robust enough so that $f$ is continuous
  (resp.~computable) if and  only if $f_{\mathsf{dec}}$ is. Cauchy
  continuity and computability also fall under this category. In
  contrast, uniform continuity and uniform computability are very
  sensitive to encoding. As an example, the function which maps a word
  to the second letter in the word is \emph{never} uniformly continuous,
  since the encoding of the first letter may be arbitrarily
  long. Nevertheless, uniform computability is
  still a relevant notion, as it provides guarantees on the maximal number of
  input data values which need to be read to produce a given number of output
  data values, even though the encoding of those values can be arbitrarily
  large.
\end{rem}

\subsection{Computability versus continuity}\label{sec:compvscont}

In this section, we show that all the computability notions
(computability, Cauchy continuity, \ldots) imply their respective
continuity notions. We then give general conditions under which the
converse also holds.

\subsubsection{From computability to continuity}
\begin{thm}\label{thm:comp2cont}
Let $f:\D^\omega\rightarrow \D^\omega$ and let $m:\N\rightarrow \N$, the following implications hold:
\begin{itemize}
  \item $f$ is computable $\Rightarrow$ $f$ is continuous
  \item $f$ is Cauchy computable $\Rightarrow$ $f$ is Cauchy continuous
  \item $f$ is uniformly computable $\Rightarrow$ $f$ is uniformly continuous
  \item $f$ is $m$-computable $\Rightarrow$ $f$ is $m$-continuous
\end{itemize}
\end{thm}

\begin{proof}
    Assume that $f$ is computable by a deterministic
    multi-tape Turing machine $M$. Let $x$ be in the topological closure
    of $\dom(f)$ and $(x_n)_n$ be a sequence in $\dom(f)$ converging to
    $x$. We show that $(f(x_n))_n$ converges to $M(x)$ if $M(x)$ is infinite, which is the case for all $x \in \dom(f)$ since $f$ is computable.
    For all
    $k\geq 0$, let $p_k$ the prefix of $x$ of length $k$. Since
    $\lim_{n\rightarrow\infty} x_n= x$, for all $k$, there exists $n_k$ such
    that for all $m\geq n_k$, $p_k\leq x_{m}$. As $M$
    is a deterministic machine, it implies that $M(x,k) \leq
    f(x_m)$. So, for all $k$, $M(x,k)\leq f(x_m)$ for all but
    finitely many $m$. It follows that $(f(x_m))_m$ converges to
    $M(x) = \lim_{k\rightarrow\infty} M(x,k)$ if $M(x)$ is an infinite word, which it is if $x\in\dom(f)$, entailing continuity of
    $f$. If additionally $f$ is Cauchy computable, then it is also the
    case for all $x\in\overline{\dom(f)}$, entailing Cauchy continuity
    of $f$.

    It remains to show the fourth statement, which entails the
    third. So, let us assume that $f$ is $m$-computable by some machine $M$. We show it is
    $m$-continuous. Let $i\geq 0$, $x,y\in\dom(f)$ such that $|x\wedge
    y|\geq m(i)$. We must show that $|f(x)\wedge f(y)|\geq i$. We have
    $M(x,m(i)) = M(y,m(i))$ because $M$ is deterministic and $|x\wedge
    y|\geq m(i)$. By
    definition of $m$-computability, we also have $|M(x,m(i))|\geq
    i$. Since $M(x,m(i))\leq f(x)$ and $M(y,m(i))\leq f(y)$, we get
    $|f(x)\wedge f(y)|\geq i$, concluding the proof.
\end{proof}

\subsubsection{From continuity to computability} While, as we have
seen in the last section, computability of a function $f$ implies its continuity, and
respectively for all the notions of computability and continuity we
consider in this paper, the converse may not hold in general. We give
sufficient conditions under which the converse holds for $f$, namely
when it has a computable next-letter problem. This problem asks, given
as input two finite words $u,v\in\D^*$, to output, if it exists, a data value $d\in
\D$ such that for all $y\in \D^\omega$ such that $uy\in\dom(f)$, we have $vd\leq f(uy)$. Because of the universal
quantification on $y$, note that $d$ is unique if it
exists. Formally:

\vspace{2mm}
\begin{minipage}{0.8\textwidth}
\underline{\textbf{Next-Letter Problem} for $f$}:

\begin{tabular}{lll}
  \textbf{Input} & $u,v\in \D^*$ \\
  \textbf{Output} & $\left\{\begin{array}{ll} d\in\D & \text{if for
                                                       all }y\in
                                                       \D^\omega\text{
                                                       s.t. }
                                                       uy\in\dom(f),
                                                        vd\leq f(uy)
 \\
                         \texttt{none} & \text{otherwise}\end{array}\right.$
\end{tabular}
\end{minipage}

\begin{thm}\label{thm:cont2comp}
  Let $f:\D^\omega\rightarrow \D^\omega$ be a function with a computable next-letter problem and let $m:\N\rightarrow \N$, the following implications hold:
  \begin{itemize}
    \item $f$ is continuous $\Rightarrow$ $f$ is computable
    \item $f$ is Cauchy continuous $\Rightarrow$ $f$ is Cauchy computable
    \item $f$ is uniformly continuous $\Rightarrow$ $f$ is uniformly computable
    \item $f$ is $m$-continuous $\Rightarrow$ $f$ is $m$-computable
  \end{itemize}
  \end{thm}

\begin{proof}
    Let us assume the existence of a procedure
    $\textsf{Next}_f(u,v)$ which computes the next-letter
    problem. To show the four statements, we show the
    existence of a deterministic Turing machine $M_f$ common to the four
    statements, in the sense that $M_f$ computes $f$
    if $f$ is continuous, respectively Cauchy computes $f$ is f is
    Cauchy continuous, etc. The Turing machine $M_f$ is best presented
    in pseudo-code as follows.

    \begin{algorithm}[H]
        \KwData{$x\in \D^\omega$}
        $v:=\epsilon$\;
        \For{$i=0$ \bf{to} $+\infty$}{
          $d := \textsf{Next}_f(x[{:}i],v)$\;
          \While{$d \neq \texttt{none}$\label{line:test}}{
            output $d$\tcp*{write $d$ on the output tape}\label{line:append}
            $v := vd$\;
            $d := \textsf{Next}_f(x[{:}i],v)$\;
          }
          \tcc{end while loop when $d = \texttt{none}$}
        }
        \caption{Turing machine $M_f$ defined in pseudo-code.}
    \end{algorithm}

    Now, we show that if $f$ is continuous, then $M_f$
    computes $f$, i.e.\ $M_f(x) = f(x)$ for all $x\in\dom(f)$. First,
    for all $i\geq 0$, let $v_i = d_{i,1}d_{i,2}\dots$ be the sequence
    of data values outputted at line~\ref{line:append} at the $i$th iteration
    of the for loop. Note that the $i$th iteration may not exist when
    the test at line~\ref{line:test} is forever true. In that case, we
    set $v_i = \epsilon$. By definition of $M_f(x,i)$, we
    have\footnote{For any $\omega$-word $\alpha$, we let
      $\alpha.\epsilon = \alpha$.}
    $M_f(x,i) = v_0v_1\dots v_i$ for all $i\geq 0$.
    Moreover, by
    definition of the next-letter problem, we also have
    $M_f(x,i)\leq f(x)$. Now, by definition, $M_f(x) =
    v_0v_1v_2\dots$. So, if it is infinite, then $M_f(x) = f(x)$. It remains to
    show that it is indeed true when $f$ is continuous. Suppose it is
    not the case. Then there exists $i_0$ such that for all $i\geq
    i_0$ the call $\textsf{Next}_f(x[{:}i],v_0\dots
    v_{i_0})$ returns $\texttt{none}$ (assume $i_0$ is the smallest
    index having this property). Let $d$ such that
    $v_0\dots v_{i_0}d< f(x)$. Then, for all $i\geq
    i_0$, there exists $\alpha_i\in\data^\omega$ and $d'\neq d$ such that
    $x[{:}i]\alpha_i\in \dom(f)$ and $v_0\dots v_{i_0} d'<
    f(x[{:}i]\alpha_i)$. Clearly, the sequence
    $(f(x[{:}i]\alpha_i))_i$, if it converges, does not converge to
    $f(x)$. Since $(x[{:}i]\alpha_i)_i$
    converges to $x$, this contradicts the continuity of $f$.

    Consider now the case where $f$ is Cauchy-continuous. It implies
    that $f$ admits a unique continuous extension $\overline{f}$ to
    $\overline{\dom(f)}$ (see Remark~\ref{rem:cauchy-cont}). By the first statement we just proved,
    $M_{\overline{f}}$ computes $\overline{f}$, and by definition of
    Cauchy-computability, we conclude that $M_f$ Cauchy-computes $f$.

    We finally prove the fourth statement, which implies the
    third. Suppose that $f$ is $m$-continuous for some modulus of
    continuity $m$. We show that $M_f$ $m$-computes $f$. We already
    proved that $M_f$ computes $f$ (as $f$ is continuous). Let $i\geq
    0$ and $x\in\dom(f)$. It remains to show that $|M_f(x,m(i))|\geq
    i$. Let $j = m(i)$. Since the algorithm $M_f$ calls $\textsf{Next}_f(x[{:}j],\cdot)$ until it
    returns $\texttt{none}$, we get that $v_0v_1\dots v_j$ is the
    longest output which can be safely output given only
    $x[{:}j]$, i.e.\ $v_0\dots v_j = \bigwedge \{ f(x[{:}j]\alpha)\mid
    x[{:}j]\alpha \in\dom(f)\}$. If $v_j$ is infinite, then
    $|M_f(x,j)| = +\infty$ and we are done. Suppose $v_j$ is
    finite. Then, there exists $\alpha\in\data^\omega$ such that
    $x[{:}j]\alpha\in\dom(f)$ and $f(x)\wedge f(x[{:}j]\alpha) =
    v_0v_1\dots v_j = M_f(x,j)$. Since $|x\wedge x[{:}j]\alpha|
    \geq j = m(i)$, by $m$-continuity of $f$, we get that
    $|M_f(x,j)|\geq i$, concluding the proof.
\end{proof}

\section{Oligomorphic data}%
\label{sec:case_oligomorphic}
In this section we consider a structure $\D$ which is oligomorphic. We will show that in this case one can decide, under reasonable computability assumptions, all the notions of continuity introduced in the previous section, as well as compute the next-letter problem.
The first step is to prove characterizations of these properties, and then show that the characterizations are decidable.
Let $\rn:\N\rightarrow \N$ denote the \emph{Ryll-Nardzewski function} of $\D$ which maps $k$ to the number of orbits of $k$-tuples of data values, which is finite thanks to oligomorphism.

\subsection{Characterizing functionality and continuity}
The main goal of this section is to give for \nrt characterizations of functionality, continuity and uniform continuity, that consist in small witnesses. These small witnesses are obtained by pumping arguments that rely on the fact that there is only a finite number of configuration orbits.

We start by defining \emph{loop removal}, a tool which will prove useful throughout this section. The main idea is that although no actual loop over the same configuration can be guaranteed over infinite runs, the oligomorphicity property guarantees that a configuration \emph{orbit} will be repeated over long enough runs. A \emph{loop} is a run of the shape $C \xrightarrow {u|v} D\xrightarrow {w|z} \mu(D)$ such that $\mu(C)=C$. The shorter run $C\xrightarrow{\mu(u)|\mu(v)}\mu(D)$ is thus called the run obtained after \emph{removing the loop}.

\begin{prop}[Small run witness]%
\label{prop:loop}
  let $T$ be an \nrt with $k$ registers and state space $Q$, and let $C\xrightarrow {u|v} D$ be a run. Then there exists $u',v'$ with $|u'|\leq |Q|\cdot\rn(2k)$ such that $C\xrightarrow {u'|v'} D$.
\end{prop}

\begin{proof}
	The idea behind this proposition is simple: any large enough run must contain a loop and can thus be shortened.
  Let $u=a_1\cdots a_n$ with a run $C_1\xrightarrow {a_1|v_1} C_2\xrightarrow {a_2|v_2} C_3\cdots C_n$.
  Let us assume that $n > |Q|\rn(2k)$, we want to obtain a shorter run from $C_1$ to $C_n$.
  Let us consider the orbits of the pairs $(C_1,C_1),(C_1,C_2),\ldots,(C_1,C_n)$. Since $n > |Q|\rn(2k)$, there must be two pairs in the same orbit, \ie there must be two indices $1\leq i<j\leq n$ and some automorphism $\mu$ such that $\mu(C_1,C_i)=(C_1,C_j)$. Hence we obtain the run $\mu(C_1)\xrightarrow {\mu(a_1\cdots a_{i-1}|v_1\cdots v_{i-1})} \mu(C_i)=C_1\xrightarrow {\mu(a_1\cdots a_{i-1}|v_1\cdots v_{i-1})} C_j$ which is strictly shorter.
\end{proof}

\subsubsection{Characterization of non-emptiness}
Let an \nra (non deterministic register automaton) be simply an \nrt without any outputs. We give a characterization of non-emptiness in terms of small witnesses.
\begin{prop}%
  \label{prop:emptiness-oligo}
Let $A$ be an \nra, the following are equivalent:
\begin{enumerate}
\item\label{itm:ONonEmpty} $A$ recognizes at least one word
\item\label{itm:OWitness} there exist $C_0 \xrightarrow{u_1} C \xrightarrow {u_2} \mu(C)$ with $C$ being a final configuration, and $\mu\in \aut(\D)$
    \item\label{itm:OSmallWitness} there exist $C_0 \xrightarrow{u_1} C \xrightarrow {u_2} \mu(C)$ with $|u_1|,|u_2|\leq |Q|\cdot\rn(2k) $, $C$ being a final configuration, and $\mu\in \aut(\D)$
\end{enumerate}
\end{prop}

\begin{proof}
Let us assume (\ref{itm:ONonEmpty}) and let $x$ be a word accepted by $A$. Since $\D$ is oligomorphic, there is only a finite number of orbits of configurations. Thus an accepting run of $A$ over $x$ must go through some accepting configuration orbit infinitely often, and in particular at least twice. Hence (\ref{itm:OWitness}) holds.

Let us assume that (\ref{itm:OWitness}) holds. To bound the size of $u_1$ and $u_2$ and obtain (\ref{itm:OSmallWitness}), we apply Proposition~\ref{prop:loop} twice in $A$ (seen as a trivial NRT producing only epsilon outputs): once to remove loops in the run $C_0 \xrightarrow{u_1 \mid \epsilon} C$ and once in $C \xrightarrow{u_2 \mid \epsilon} \mu(C)$.

Let us assume (\ref{itm:OSmallWitness}), then the word $u_1u_2\mu(u_2)\mu^2(u_2)\cdots$ is accepted by $A$.
\end{proof}

\subsubsection{Characterizing functionality}
Characterizing functionality is slightly more complicated since we have to care about outputs. Moreover, we need to exhibit patterns involving two runs which makes pumping more involved.

We start with a useful yet quite technical lemma, which will allow us to remove loops while preserving mismatches.
\begin{lem}[Small mismatch witness]%
  \label{lem:smallmis}
  There exists a polynomial $P(x,y,z)$ such that the following hold.

  Let $T$ be an \nrt, with $k$ registers, a state space $Q$ and a maximal output length $L$.
  Let $C_1\xrightarrow{u|u_1}D_1$ and $C_2\xrightarrow {u|u_2}D_2$ be runs such that $u_1\mismatch u_2$.
  Then there exists $u'$ of length less than $P(\rn(4k+2),|Q|,L)$ and $u_1',u_2'$, so that $C_1\xrightarrow{u'|u_1'}D_1$, $C_2\xrightarrow {u'|u_2'}D_2$ with $u_1'\mismatch u_2'$.

\end{lem}

\begin{proof}
  Let $\rho_1=C_1\xrightarrow{u|u_1}D_1$ and $\rho_2=C_2\xrightarrow {u|u_2}D_2$ be runs such that $u_1\mismatch u_2$.

  Our goal is to show that either $|u|$ is smaller than some polynomial in $\rn(4k+2),|Q|$ and $L$ or we can remove a synchronous loop in $\rho_1,\rho_2$ while preserving the mismatch.
  A synchronous loop is defined as a loop over the same input in the product transducer $T\times T$.
  Let $u_1=\alpha a_1\beta_1$, $u_2=\alpha a_2\beta_2$ with $a_1\neq a_2$.
  We will only consider removing loops that do \emph{not} contain the transitions producing the mismatching letters $a_1$ or $a_2$.

In order to make things more precise we introduce  some useful loop vocabulary.
   An \emph{initial loop} in $\rho=C_1\rightarrow C_2\rightarrow \cdots \rightarrow C_n$ is a loop starting in $C_1$.
  A \emph{simple loop} is a loop $C\rightarrow D\rightarrow \mu(D)$ so that there is no pair of configurations $B,\lambda(B)$ with $\lambda(C)=C$ which occur within the part $D\rightarrow \mu(D)$ of the run. Similarly a \emph{simple synchronous loop} is a simple loop over $T\times T$.

In the following we will consider loops that are synchronous and simple.
They will be of the shape $(C_1,C_2)\rightarrow (B_1,B_2) \rightarrow \mu(B_1,B_2)$ with $\mu(C_1,C_2)=(C_1,C_2)$. Moreover, we will ask that the automorphism satisfies $\mu(a_1)=a_1$ and $\mu(a_2)=a_2$, this is done to be sure that removing or adding loops does not affect the fact that $a_1\neq a_2$ (note that it would be enough to simply ask that $\mu(a_1)=a_1$). We call such loops \emph{nice}.
 With these constraints, if we have a sequence of configuration pairs in $(\rho_1,\rho_2)$ of length larger than $\rn(4k+2)Q^2$, we are sure to find a nice loop that preserves $(a_1,a_2)$.

  We call the \emph{effect} of a loop on $\rho_1$, the length of the output factor removed from $\alpha$, in particular the effect is $0$ if the factor removed is in $\beta_1$ (and symmetrically for $\rho_2$).
  We call the \emph{net effect} of a synchronous loop on $\rho_1,\rho_2$ the difference between the effect in $\rho_1$ and the effect on $\rho_2$. Loops with a positive, negative and null net effect are called positive, negative and null loops, respectively.

  If there exists a null nice loop, then removing it preserves the mismatch and we are done.
  We split the proof into the two remaining cases: (1)~either all nice loops are strictly positive (without loss of generality) or (2)~some nice loops are strictly positive while others are strictly negative. In the first case, we will show that $u$ has to be \emph{small}. In the second case, we will show that, after removing all nice loops, we can repump (a small number of) positive and negative nice loops to realign the mismatch.

  Let us consider a configuration $(B_1,B_2)$ in $(\rho_1,\rho_2)$.
  The \emph{effects} of $(B_1,B_2)$ is the set of effects of nice loops starting in $(C_1,C_2)$ ending in $(B_1,B_2)$ (note that we consider any nice loop, not just the ones occurring in $(\rho_1,\rho_2)$).
  Then we say that $(B_1,B_2)$ is \emph{positive} (resp.~\emph{negative}) if all its effect are positive (resp.~\emph{negative}).

  We consider two cases, (1)~either $(\rho_1,\rho_2)$ only has strictly positive configurations (without loss of generality) or (2)~it has configurations with positive and negative effects (possibly different configurations).

  Let us start with the first case, we observe two simple facts: the effect of a nice loop is bounded by $\rn(4k+2)|Q|^2L$ and, similarly, the output length of a run without nice loops is bounded by $\rn(4k+2)|Q|^2L$. Let us denote by $\Delta_0$ the difference of output lengths $|u_2|-|u_1|$. Since there are no nice loops having $0$ effect on $\rho_1$ this means that $|\beta_1|\leq \rn(4k+2)|Q|^2L$, and thus $\Delta_0\geq -\rn(4k+2)|Q|^2L$.
  We denote by $\Delta_1$ the difference of output length after removing one nice loop from $(\rho_1,\rho_2)$. We observe that $-\rn(4k+2)|Q|^2L\leq \Delta_0 < \Delta_1$ since all nice loops must be strictly positive. We carry on removing nice loops until there are none left.
Note that removing nice loops cannot introduce configurations with new effects: some configurations are erased, and to some an automorphism $\mu$ satisfying $\mu(C_1,C_2,a_1,a_2)=(C_1,C_2,a_1,a_2)$ is applied, which preserves the effect set.
We thus obtain $-\rn(4k+2)|Q|^2L\leq \Delta_0 <\Delta_1 < \ldots <\Delta_l \leq \rn(4k+2)|Q|^2L$ after removing $l$ nice loops and obtaining a run without nice loops. Hence we get that $l\leq 2\rn(4k+2)|Q|^2L$. This means that the runs $\rho_1,\rho_2$ could not have been very large. In fact we have $|\rho_1|\leq (l +1)\times \rn(4k+2)|Q|^2\leq 3\rn(4k+2)^2|Q|^4L$.

  We only have left to consider the second case, that is, $(\rho_1,\rho_2)$ has configurations with positive and negative effects. Note that the effects of configurations belong to the interval $[-\rn(4k+2)|Q|^2L,\rn(4k+2)|Q|^2L]$.
  Let $d$ denote in the following the $\gcd$ of all effects of configurations in $(\rho_1,\rho_2)$.
Let $(B_1^1,B_2^1), \ldots,(B_1^l,B_2^l)$ be configurations so that the $\gcd$ of their collective effects is $d$, and with at least one positive and one negative effect. We can assume that $l\leq \rn(4k+2)|Q|^2L$. We remove nice loops without deleting these configurations. Note that removing loops may mean applying an automorphism to some of these configurations. However, the automorphisms always preserve $(C_1,C_2,a_1,a_2)$ so the effect set of the configuration is left unchanged. Also note that the effects of such loops are all multiples of $d$.
After removing all these loops we are left with runs that are small ($\leq (l+1)\rn(4k+2)|Q|^2$) but the outputs $a_1$, $a_2$ may very well be misaligned (by a factor of $d$). The idea is to use the configurations $(B_1^1,B_2^1), \ldots,(B_1^k,B_2^k)$ to pump simple loops into the run to realign the two mismatching outputs. Let us explain how these nice loops can be pumped into the runs:
Let $\alpha:(C_1,C_2)\xrightarrow{(u_1,u_2)|(v_1,v_2)} (B_1,B_2)$ be a run and let $\beta:(C_1,C_2)\xrightarrow{(x_1,x_2)|(y_1,y_2)} (B_1,B_2)\xrightarrow {(z_1,z_2)|(w_1,w_2)}\mu(B_1,B_2)$ be a nice loop. By applying $\mu^{-1}$ to $\alpha$ and the second part of $\beta$ we obtain a run of the shape:
$(C_1,C_2)\xrightarrow {\mu^{-1}(u_1,u_2)|\mu^{-1}(v_1,v_2)}\mu^{-1}(B_1,B_2)\xrightarrow {\mu^{-1}(z_1,z_2)|\mu^{-1}(w_1,w_2)}(B_1,B_2)$, with the added effect of $\beta$.

Note that we have chosen the configurations carefully so that we can change the alignment by any multiple of $d$.
The next lemma (whose proof can be found in Appendix~\ref{sec:bezout})
shows that it is possible to find small iteration numbers for
each loop so as to realign the outputs:
  \begin{lem}[Signed generalized B\'ezout identity]%
    \label{lem:bezout}
  Let $p_1,\ldots,p_k\in \Z{\setminus}{\set{0}}$ be non-zero integers, such that at least two have different signs.
  Then there exist natural numbers $n_1,\ldots, n_k \leq \max(|p_1|^3,\ldots,|p_k|^3)$ such that:
  \[n_1p_1+\ldots +n_{k}p_k=\gcd(p_1,\ldots,p_k)\]
  \end{lem}

  Using this lemma, we can change the alignment by $d$ using only a polynomial number of times each loop (at most $(\rn(4k+2)|Q|^2L)^3$ for each loop). Since the runs are small, the misalignment is also small (at most $(l+1)\rn(4k+2)|Q|^2L$) and we only need to repeat this operation a polynomial number of times.
  Finally since we have chosen automorphisms that preserve $a_1,a_2$, we are sure to obtain a mismatch.

$\hfill  \textit{End of Proof of Lemma~\ref{lem:smallmis}}$

\end{proof}

\begin{prop}[Functionality]%
\label{prop:func}
  There exists a polynomial $P(x,y,z)$ such that the following holds.

  Let $R\subseteq \D^\omega \times \D^\omega$ be given by an \nrt $T$ with $k$ registers, a state space $Q$ and a maximum output length $L$. The following are equivalent:
  \begin{enumerate}
    \item $R$ is not functional
    \item there exist $C_0 \xrightarrow{u|u_1} C_1 \xrightarrow {v|v_1} D_1 \xrightarrow{w|w_1} \mu(C_1)$, $C_0 \xrightarrow{u|u_2} C_2 \xrightarrow {v|v_2} D_2 \xrightarrow{w|w_2} \mu(C_2)$ with $C_1,D_2$ final, $\mu\in \aut(\D)$ such that $u_1\mismatch u_2$
    \item there exist $C_0 \xrightarrow{u|u_1} C_1 \xrightarrow {v|v_1} D_1 \xrightarrow{w|w_1} \mu(C_1)$, $C_0 \xrightarrow{u|u_2} C_2 \xrightarrow {v|v_2} D_2 \xrightarrow{w|w_2} \mu(C_2)$ with $C_1,D_2$ final and $|u|,|v|,|w|\leq P(\rn(4k),|Q|,L) $, $\mu\in \aut(\D)$ such that $u_1\mismatch u_2$

  \end{enumerate}
  \end{prop}

\begin{proof}
Let us assume that~(1) holds, meaning that $T$ has two accepting runs $\rho_1,\rho_2$ over some word $x\in \D^\omega$ which produce different outputs. Let $C_0, C_1,C_2 \ldots$ denote the configurations of $\rho_1$ and $C_0, C_1',C_2' \ldots$ the ones of $\rho_2$. Let us consider the orbits of pairs of configurations $C_i,C_i'$. We know that there is a finite number of such orbits. We also know that an infinite number of such pairs is accepting in the first component and an infinite number is accepting in the second component. Thus we can see $(\rho_1,\rho_2)$ as a sequence of the following shape, with $C$ and $D'$ final:

 \[ (C_0,C_0)\xrightarrow{u_0}(C,C')\xrightarrow {v_0}(D,D')\xrightarrow{u_1} {\mu_1(C,C')} \xrightarrow {v_1}\nu_1(D,D') \xrightarrow {u_2}\cdots\]

Since the outputs are different and infinite then they mismatch at some position $i$. Then, there exists $n$ such that $u=u_0v_0\cdots u_n v_n$ has produced at least $i$ symbols, both for $\rho_1$ and $\rho_2$.
Hence we have shown that $C_0 \xrightarrow{u|\alpha_1} \mu_n(C) \xrightarrow {u_{n+1}|\beta_1} \nu_{n+1}(D)\xrightarrow {v_{n+1}|\gamma_1} \mu_{n+1}(C)$, $C_0 \xrightarrow{u|\alpha_2} \mu_n(C') \xrightarrow {u_{n+1}|\beta_2} \nu_{n+1}(D')\xrightarrow {v_{n+1}|\gamma_2} \mu_{n+1}(C')$, and by assumption $|\alpha_1|,|\alpha_2|\geq i$, and thus $\alpha_1\mismatch \alpha_2$.
Hence we have shown that~(2) holds.

Let us assume that~(2) holds: \ie we have  $C_0 \xrightarrow{u|u_1} C_1 \xrightarrow {v|v_1} D_1 \xrightarrow{w|w_1} \mu(C_1)$, $C_0 \xrightarrow{u|u_2} C_2 \xrightarrow {v|v_2} D_2 \xrightarrow{w|w_2} \mu(C_2)$ with $C_1,D_2$ final, $\mu\in \aut(\D)$ such that $u_1\mismatch u_2$.
We use Lemma~\ref{lem:smallmis} (small mismatch witness) to obtain that $C_0 \xrightarrow{u'|u_1'} C_1$, $C_0 \xrightarrow{u'|u_2'} C_2 $ with $u_1'\mismatch u_2'$ and $|u'|$ small. We get
$C_0 \xrightarrow{u'|u_1'} C_1 \xrightarrow {v|v_1} D_1 \xrightarrow{w|w_1} \mu(C_1)$, $C_0 \xrightarrow{u'|u_2'} C_2 \xrightarrow {v|v_2} D_2 \xrightarrow{w|w_2} \mu(C_2)$.
We now only have to use some loop removal on $v$ and $w$, just as in Proposition~\ref{prop:emptiness-oligo}, in order to obtain words smaller than $|Q|^2\rn(4k)$.

Showing that~(3) implies~(1) is the easiest part since the pattern clearly causes non-functionality over the word $uvw\mu(vw)\mu^2(vw)\cdots$.
\end{proof}

\subsubsection{Characterizing continuity}

Here we characterize continuity and uniform continuity using patterns similar to the one of functionality.
Before doing so we introduce a property of configuration: a configuration is \emph{co-reachable} if there is a final run from it.
%A configuration is \emph{constant} if all the final runs from it produce the same infinite word.

For this we define a notion of \emph{critical pattern}.
\begin{defi}%
  \label{def:criticalPattern}
Let $T$ be an \nrt. A pair of runs of the form
 $C_0 \xrightarrow{u|u_1} C_1 \xrightarrow {v|v_1} \mu(C_1)\xrightarrow{w|w_1}D_1$, $C_0 \xrightarrow{u|u_2} C_2 \xrightarrow {v|v_2} \mu(C_2)\xrightarrow{z|w_2}D_2$
  is a \emph{critical pattern} if $D_1$, $D_2$ are co-reachable and one of the following holds:
 \begin{enumerate}[(a)]
    	  \item\label{itm:critMismatch} $u_1\mismatch u_2$, or
      %\item $v_i=\epsilon$ and $C_i$ is not constant for some $i\in \set{1,2}$, or
      \item\label{itm:critOneEmpty} $v_i=\epsilon$ and $u_j \mismatch  u_{i}w_i$ for $\set{i,j}= \set{1,2}$, or
      \item\label{itm:critTwoEmpty} $v_1=v_2=\epsilon$ and $u_1w_1\mismatch u_2w_2$
    \end{enumerate}
We  denote by
 $\cp_T(u,v,w,z,C_1,\mu(C_1),D_1,C_2,\mu(C_2),D_2)$ ($T$ is omitted when clear from context)
 the set of critical patterns.
  \end{defi}

Before characterizing continuity and uniform continuity, we show a small critical pattern property.
\begin{clm}[Small critical pattern]%
  \label{clm:smallmis}
  There exists a polynomial $P'(x,y,z)$ such that the following holds.

  Let $T$ be an \nrt, with $k$ registers, a state space $Q$ and a maximal output length $L$.
  Let $C_0 \xrightarrow{u|u_1} C_1 \xrightarrow {v|v_1} \mu(C_1)\xrightarrow{w|w_1}D_1$, $C_0 \xrightarrow{u|u_2} C_2 \xrightarrow {v|v_2} \mu(C_2)\xrightarrow{z|w_2}D_2$ be a critical pattern.
  Then there exists $u',v',w',z'$ of length less than $P'(\rn(4k),|Q|,L)$ and $u_1',v_1',w_1',u_2',v_2',w_2'$, so that $C_0 \xrightarrow{u'|u_1'} C_1 \xrightarrow {v'|v_1'} \mu(C_1)\xrightarrow{w'|w_1'}D_1$, $C_0 \xrightarrow{u'|u_2'} C_2 \xrightarrow {v'|v_2'} \mu(C_2)\xrightarrow{z'|w_2'}D_2$ is a critical pattern.
\end{clm}

\begin{proof}
%Let $C_0 \xrightarrow{u|u_1} C_1 \xrightarrow {v|v_1} \mu(C_1)\xrightarrow{w|w_1}D_1$, $C_0 \xrightarrow{u|u_2} C_2 \xrightarrow {v|v_2} \mu(C_2)\xrightarrow{z|w_2}D_2$ be a critical pattern.
%Using Proposition~\ref{prop:loop}, we can make sure that $v$ is small since smaller outputs $v_1,v_2$ do not affect the existence of a critical pattern.
We want to remove loops in $u,v,w,z$ without affecting the mismatches.
The idea is to see such a critical pattern as a pair of runs which mismatch and leverage Lemma~\ref{lem:smallmis}. In order to make sure that the loops which are removed do not interfere with the intermediate configurations, we color the states of the transducer. We consider runs which start with red configurations, then the middle parts $C_1 \xrightarrow {v|v_1} \mu(C_1)$ and $C_2 \xrightarrow {v|v_2} \mu(C_2)$ are colored in blue and the final parts are colored in green. Using Lemma~\ref{lem:smallmis}, we obtain runs smaller than $P(\rn(4k),3|Q|,L)$ (the $3$ factor comes from the coloring). Since the loops have to be removed in the monochromatic parts, we obtain the desired result.
\end{proof}

Let us give a characterization of continuity and uniform continuity for functions given by an \nrt.
\begin{prop}[Continuity/uniform continuity]%
  \label{prop:crit}
  There exists a polynomial $P'(x,y,z)$ such that the following holds.
  Let $f:\D^\omega \rightarrow \D^\omega$ be given by an \nrt $T$ with $k$ registers, a state space $Q$ and a maximum output length $L$. The following are equivalent:
  \begin{enumerate}
    \item $f$ is not uniformly continuous (resp.~continuous)
    \item there exists a critical pattern in $\cp(u,v,w,z,C_1,\mu(C_1),D_1,C_2,\mu(C_2),D_2)$ (resp.~with $C_1$ final).
    \item there exists a critical pattern in $\cp(u,v,w,z,C_1,\mu(C_1),D_1,C_2,\mu(C_2),D_2)$ such that $|u|,|v|,|w|,|z|\leq P'(\rn(4k),|Q|,L) $ (resp.~with $C_1$ final).
  \end{enumerate}
  \end{prop}

\begin{proof}
  Let us assume that~(1) holds, meaning that $f$ is not uniformly continuous (resp.\ not continuous) at some point $x\in \D^\omega$. This means that there exists $i$ such that for any $n$, there are two accepting runs $\rho_1,\rho_2$ over $x_1,x_2$ and producing $y_1,y_2$ respectively such that $|x\wedge x_1\wedge x_2|>n$ and $|y_1\wedge y_2|<i-1$. Moreover, if $f$ is not continuous, we can even assume that $x_1=x$ and $\rho_1=\rho$, some accepting run over $x$. Let us consider some $n>2i\cdot |Q|^2\cdot\rn(2k)$. Let $u=x\wedge x_1\wedge x_2$, since $u$ is large enough we have that some pair of configurations in $\rho_1,\rho_2$ has to repeat at least $2i$ times, up to automorphism. If $f$ is not continuous, we choose $n$ large enough so that a final configuration appears at least $2i\cdot |Q|^2\cdot\rn(2k)$ times in the first $n$ transitions of $\rho$. That way we can ensure that some pair of configuration in $\rho,\rho_2$ repeats at least $2i$ times, up to automorphism, with the configuration of $\rho$ being accepting.

Thus we obtain two sequences: $C_0\xrightarrow {u_0|v_0} C\xrightarrow{u_1|v_1} \mu_1(C)\cdots \mu_{{2i-1}}(C)\xrightarrow{u_{2i}|v_{2i}} \mu_{2i}(C)$ and  $C_0\xrightarrow {u_0|w_0} D\xrightarrow{u_1|w_1} \mu_1(D)\cdots \mu_{{2i-1}}(D)\xrightarrow{u_{2i}|w_{2i}} \mu_{2i}(D)$. Moreover, let $x_1=ux_1'$, $x_2=ux_2'$, let $y_1=v_0\cdots v_{2i}y_1'$ and $y_2=w_0\cdots w_{2i}y_2'$.
We do a case analysis; note that the cases are not necessarily mutually exclusive.

Case~(1) let us first assume that there is some index $j\in \set{1,\ldots,2i}$ so that $\mu_{j}(C)\xrightarrow {u_j|\epsilon}\mu_j(C)$ and $\mu_{j-1}(D)\xrightarrow{u_i|\epsilon}\mu_j(D)$.
Since the outputs $y_1,y_2$ mismatch, we have some prefixes of $\rho_1,\rho_2$ of the shape $C_0\xrightarrow{u_0\cdots u_{j-1}|v_0\cdots v_{j-1}}\mu_{j}(C)\xrightarrow {u_j|\epsilon}\mu_j(C)\xrightarrow {x_1''|y_1''}E$ and $C_0\xrightarrow{u_0\cdots u_{j-1}|w_0\cdots w_{j-1}}\mu_{j}(D)\xrightarrow {u_j|\epsilon}\mu_j(D)\xrightarrow {x_2''|y_2''}F$ with $v_0\cdots v_{j-1}y_1''\mismatch w_0\cdots w_{j-1}y_2''$. Hence we obtain a critical pattern of shape (c).

Case~(2) we assume that there are at least $i$ indices $j\in \set{1,\ldots,2i}$ such that $v_j\neq \epsilon$ and at least $i$ indices such that $w_j\neq \epsilon$. Then we have that $v_0\cdots v_{2i}\mismatch w_0\cdots w_{2i}$ and thus we get a critical pattern of shape (a).

Case~(3), let us assume (without loss of generality) that there strictly fewer than $i$ indices such that $w_j\neq \epsilon$. This means that ther must be at least $i$ indices such that $v_j\neq \epsilon$, otherwise we refer back to case~(1). Let us consider a prefix of $\rho_2$ $C_0\xrightarrow{u_0\cdots u_{2i}|w_0\cdots w_{2i}}\mu_{2i}D\xrightarrow {x_2''|y_2''}$ such that $v_0\cdots v_{2i}\mismatch  w_0\cdots w_{2i}y_2''$.
Let $j$ be such that $w_j=\epsilon$, then we add the loop corresponding to index $j$ after the configuration $(\mu_{2i}(C),\mu_{2i}(D)) $. Doing this may modify the mismatching letter of $y_2''$ as to cancelling the mismatch, however if it is the case, we can simply add the loop twice, which guarantees that the mismatch is preserved. Thus we obtain a critical pattern of shape (b).

If we assume that~(2) holds then Claim~\ref{clm:smallmis} gives us~(3).

Let us assume that~(3) holds. Then $f$ is discontinuous at $x=uv\mu(v)\mu^2(v)\cdots\in \overline{\dom(f)}$ and is thus not uniformly continuous. Moreover, if $C_1$ is final we have that $f$ is discontinuous at $x\in \dom(f)$, and hence $f$ is not continuous.
\end{proof}

\subsection{Deciding functionality, continuity and computability}
We use a key property of oligomorphic structures, namely that orbits can be defined by first order formulas.

Let $\D = (D,\Sigma^\D)$ be an oligomorphic structure.
We denote by
$\fo[\Sigma]$ the set of first-order formulas over signature
$\Sigma$, and just $\fo$ if $\Sigma$ is clear from the context.
\begin{propC}[{\cite[Lemma 4.11]{Bojanczyk19}}]%
  \label{prop:oligoFO}
Let $\D$ be an oligomorphic structure and let $k$ be a natural number. Any orbit of an element of $\D^k$ is first-order definable.
\end{propC}
We say that
$\D$ is \emph{decidable} if its Ryll-Nardzewski function is computable
and $\fo[\Sigma]$ has decidable satisfiability problem over $\D$. Moreover, we say that $\D$ is \emph{polynomially decidable} if an orbit of $\D^k$ can be expressed by an \fo formula of polynomial size in $k$ and the \fo satisfiability problem is decidable in  \PSpace.
One (but not us) could easily define a similar notion of exponentially decidable, or $f$-decidable for some fixed complexity function $f$.

Roughly speaking the main automata problems which we will consider (emptiness, functionality, continuity, \etc) will be \PSpace (resp.~decidable) whenever the structure $\D$ is polynomially decidable (resp.~decidable).

\subsubsection{Computing the next letter}

In this section, we show how to compute the next letter problem for
\nrt over a decidable representable oligomorphic structure $\D$. By Theorems~\ref{thm:cont2comp}
and~\ref{thm:comp2cont}, this entails that continuity and computability coincide
for functions defined by transducers over decidable representable oligomorphic
structures, as stated in Theorem~\ref{thm:compcont} below.

Before tackling the next letter problem, we consider the emptiness problem of register automata.
\begin{thm}
Let $\D$ be a decidable (resp.~polynomially decidable) oligomorphic structure.
The emptiness problem for \nra is decidable (resp.~in \PSpace).
\end{thm}

\begin{proof}
We show the result in case of a polynomially decidable structure, the more general case can be obtained by forgetting about complexity.
Let $\D$ be a polynomially decidable oligomorphic structure and let $T$ be an \nra with $k$ registers and state space $Q$. Since $\D$ is polynomially decidable, any orbit of $\D^k$ can be represented by a formula of size polynomial in $k$. This means that $\rn(k)$ is exponential.
Using the non-emptiness characterization from Proposition~\ref{prop:emptiness-oligo}, we only need to find a run of length polynomial in $Q$ and $\rn(k)$. The idea is to use a counter bounded by this polynomial, using space in $\log(|Q| \rn(k))$, and execute an \NPSpace algorithm which will guess a run of the automaton and update the type of configurations in space polynomial in $k$ and $\log(|Q|)$.
Simulating the run goes like this: first the type of the configuration is initialized to $q_0,(d_0,\ldots,d_0)$. Then a new type $\phi(y_1,\ldots,y_k)$ is guessed, as well as a transition which we see as given by a state $q$ and a formula $\psi(x_1,\ldots,x_k,y_1,\ldots,y_k)$. We check that the transition is valid by deciding the satisfiability of $\exists y_1,\ldots,y_k\ \psi(d_0,\ldots,d_0,y_1,\ldots,y_k)\wedge \phi(y_1,\ldots,y_k)$ which can be done in \PSpace, by assumption. We thus move to the new configuration type given by $q,\phi$ and we continue the simulation. At some point when $q$ is final, we keep the configuration type in memory and guess that we will see it again.
\end{proof}

\begin{lem}\label{lem:nextletter}
  Let $f : \D^\omega\rightarrow \D^\omega$ be a function defined by an
  $\nrt$ over a decidable representable oligomorphic structure $\D$. Then, its next
  letter problem is computable.
\end{lem}

\begin{proof}
In the next letter problem we get as input two words $u,v\in \D^*$.
Our first goal is to decide whether there exists $d\in D$ such that $f(u\D^\omega)\subseteq vd\D^\omega$. We check the negation of this property, i.e.\ we try to exhibit two runs $C_0\xrightarrow{u|u_1} C_1\xrightarrow{w|w_1} D_1$, $C_0\xrightarrow{u|u_2} C_2\xrightarrow{w|w_2} D_2$ such that $|u_1w_1|,|u_2w_2|> |v|$ and either $|u_1w_1\wedge v|<|v|$ or $|u_1w_1\wedge u_2w_2| \leq |v|$.
The non-existence of such runs only depends on the type of $u,v$, hence we can define an automaton which simulates $T$ and starts by reading some input of the type of $u$ and checks whether there is a mismatch occurring before the $|v|$ outputs, which can be done by adding one register to store the mismatch, and two $|v|$-bounded counters in memory (recall that $v$ is given as input). It finally checks that the reached configurations are co-reachable by guessing some continuation for each and simulating $T$ over it. Thus we reduce the non-existence of a next letter to the non-emptiness of an automaton, which is decidable.

Once we know that such a next letter exists, we only have to simulate any run of $T$ over $uw$ for an arbitrary $w \in \D^\omega$ such that $uw \in \dom(T)$, and see what the $|v|+1^{\text{th}}$ output is (note that we can avoid $\epsilon$-producing loops, so this data value will be output in a finite number of steps). To be able to simulate $T$ over $uw$, we use the fact that $\D$ is representable and decidable. For every transition that we want to take, from decidability we can check whether the transition is possible. Then, once we know the transition is possible, we can enumerate the representations of elements of $\D$ and check that they satisfy the transition formula.
\end{proof}

As a direct corollary of Lemma~\ref{lem:nextletter}, Theorem~\ref{thm:cont2comp}
and Theorem~\ref{thm:comp2cont}, we obtain:

\begin{thm}\label{thm:compcont}
    Let $f : \D^\omega\rightarrow \D^\omega$ be a function defined by an
  $\nrt$ over a decidable oligomorphic structure $\D$, and let
  $m:\N\rightarrow \N$ be a total function. Then,
  \begin{enumerate}
  \item $f$ is computable iff $f$ is continuous
  \item $f$ is uniformly computable iff $f$ is uniformly continuous
  \item $f$ is $m$-computable iff $f$ is $m$-continuous
  \end{enumerate}
\end{thm}

\subsubsection{Deciding functionality, continuity and computability}

\begin{thm}\label{thm:decoligo}
Given a decidable (resp.~polynomially decidable) oligomorphic structure $\D$ functionality, continuity and uniform continuity are decidable (resp. \PSpace-c) for functions given by \nrt. As a consequence, if $\D$ is representable, then computability and uniform computability are decidable (resp.~\PSpace-c).
\end{thm}

\begin{proof}
The proofs are very similar, whether we consider functionality, continuity or uniform continuity.
Let us show the result for functionality.
Moreover we assume that $\D$ is polynomially decidable, the argument in the more general case can easily be obtained just by forgetting about complexity.

Let us consider an \nrt $T$ with $k$ registers, state space $Q$ and maximum output length $L$. We want to show that we can decide the existence of two runs with a mismatch.
From the characterization given in Proposition~\ref{prop:func}, we know that the pattern we are looking for is small. We consider a counter bounded by the value $3P(\rn(4k),|Q|,L)$, which can be represented using polynomial space because $\rn(4k)$ is exponential in $k$ ($\D$ is polynomially decidable). Our goal is to simulate $T$ and exhibit a pattern of length bounded by that counter.
As we have seen before, we can easily simulate runs of $T$ in \PSpace. The additional difficulty here is that at some point we have to check that two output positions mismatch. We use two additional counters which will ensure that the two mismatching outputs correspond to the same position.
Let us now describe how the algorithm goes, in a high-level manner.
We start by initializing two runs in parallel, as well as our counters.
We keep in memory the $2k$-type of the two configurations, which can be done in polynomial space since $\D$ is polynomially decidable.
We keep guessing in parallel two transitions for our runs and updating the $2k$-type using the fact that satisfiability of \fo is in \PSpace. Every time a run outputs some letter, its counter is incremented.
At some point we may guess that we output the mismatching value in one of the runs, in which case we stop the counter corresponding to that run. We crucially also need to be able to check later that the value output mismatches. In order to do this we keep in memory a $2k+1$-type, always keeping the value which we output. At some point we output the second mismatching position, we check that the counters coincide and that the outputs are indeed different, which is given by the $2k+1$-type. In parallel, we also have to check that we reach some final configurations and that some configuration repeats. To do this we need to keep one or two additional $2k$-type in memory, which again can easily be done in \PSpace.

The approach for continuity and uniform continuity is exactly the same except that the patterns of Proposition~\ref{prop:crit} are slightly more involved. Moreover we also need to decide on the fly that the configurations reached are co-reachable. This can be done exactly like for non-emptiness of automata in \PSpace.

Finally, the \PSpace lower bound is obtained by reducing from emptiness of register automata over $(\N, \{=\})$, which is \PSpace-c~\cite{DBLP:journals/tocl/DemriL09}. Since the data domains are countable, they can always simulate $(\N,\{=\})$, and the proofs of~\cite{DBLP:conf/fossacs/ExibardFR20} can easily be adapted.
\end{proof}
Thus, given a specification represented as a \nrt over a representable and decidable domain, one can examine whether it can be implemented (provided it is functional, which is then decidable) by checking whether it is computable. Indeed, computability is the most liberal criterium for being realisable. The notion of uniform computability then allows to refine this check, as for functions that are uniformly computable, one can focus on implementations that have a bounded lookahead. In the case of $(\Q,\set{<})$, we further get that both problems are $\PSpace$-complete:
\begin{thm}
  For relations given by \nrt over $(\Q,\set{<})$ deciding functionality, continuity/computability and uniform continuity/uniform computability are \PSpace-complete.
\end{thm}
\begin{proof}
We only need to argue that $(\Q,\set{<})$ is representable and polynomially decidable.
Clearly rational numbers are representable.
Moreover, since $(\Q,\set{<})$ is homogeneous, it admits quantifier elimination, which means that any $k$-type can be defined by a formula polynomial in $k$ (actually linear). Indeed a type is just given by the linear order between the free variables.
Moreover, satisfiability of first-order logic over $(\Q,\set{<})$ is in \PSpace.
\end{proof}
Note that the same reasoning applies for $(\N,\set{=})$; the case of $(\N,\set{=})$ can also be obtained by encoding it in $(\Q,\set{<})$.

\section{\texorpdfstring{A non oligomorphic case: $(\N,\{<,0\})$}{A non oligomorphic case: (N,\{<,0\})}}%
\label{sec:case_N}

We now turn to the study of the case of natural numbers equipped with the usual order. This domain is not oligomorphic (cf Example~\ref{exa:oligomorphic}), so there might not exist loops in the transducer, as there are infinitely many orbits of configurations.
\begin{nota}
For simplicity, in the rest of this section, $(\N,\{<,0\})$ and $(\Q_+,\{<,0\})$ (where $\Q_+$ is the set of rational numbers which are greater than or equal to $0$) are respectively denoted $\N$ and $\Q_+$.
\end{nota}
We thus need to study more precisely which paths are iterable, i.e.\ can be repeated unboundedly or infinitely many times. We show that such property only depends on the relative order between the registers, i.e.\ on the type of the configurations, seen as belonging to $\Q_+$, where the type of a configuration is an FO-formula describing its orbit (cf Proposition~\ref{prop:oligoFO} and Section~\ref{sec:Q-types}). More generally, the fact that $\N$ is a subdomain of $\Q_+$, along with the property that any finite run in $\Q_+$ corresponds to a run in $\N$ (by multiplying all involved data values by the product of their denominator), allows us to provide a characterisation of continuity and uniform continuity which yields a \textsc{PSpace} decision procedure for those properties.

\subsection{On loop iteration}

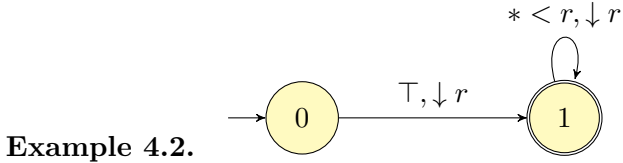
\begin{exa}%
  \label{exa:emptyN-nonemptyQ-decr}
\begin{tikzpicture}[->, >=stealth', auto, node distance=2.5cm]
  \tikzstyle{every state}=[text=black,fill=yellow!30]

  \node[state, initial, initial text={}] (0) {$0$};
  \node[state, accepting, right= of 0] (1) {$1$};

  \path (0) edge node[above] {$\top, \downarrow{} r$} (1);
  \path (1) edge[loop above] node[above] {$ * < r, \downarrow{} r$} (1);
\end{tikzpicture}
\end{exa}

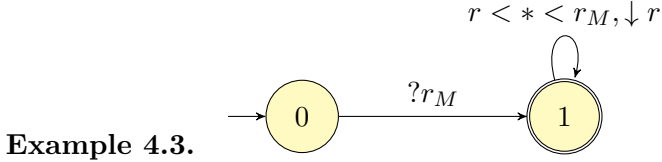
\begin{exa}%
  \label{exa:emptyN-nonemptyQ-incr-bounded}
\begin{tikzpicture}[->, >=stealth', auto, node distance=2.5cm]
  \tikzstyle{every state}=[text=black,fill=yellow!30]

  \node[state, initial, initial text={}] (0) {$0$};
  \node[state, accepting, right= of 0] (1) {$1$};

  \path (0) edge node[above] {$?r_M$} (1);
  \path (1) edge[loop above] node[above] {$r < * < r_M, \downarrow{} r$} (1);
\end{tikzpicture}
\end{exa}

The \nra of Example~\ref{exa:emptyN-nonemptyQ-decr} is non-empty in $\Q_+$, since it accepts e.g.\ the word $1 \cdot \dfrac{1}{2} \cdots \dfrac{1}{n} \cdots$. However, it is empty in $\N$. Indeed, any data word compatible with its only infinite run necessarily forms an infinite descending chain, which is impossible in $\N$. Similarly, in Example~\ref{exa:emptyN-nonemptyQ-incr-bounded}, the \nra initially guesses some upper bound $B$ which it stores in $r_M$, and then asks to see an infinite increasing chain which is bounded from above by $B$. This is possible in $\Q_+$, but not in $\N$.

That is why we need to study more closely what makes a given path \emph{$\omega$-iterable}, i.e.\ that can be taken an infinite number of times. To characterise continuity, we will also need the weaker notion of \emph{iterable} path, i.e.\ of a path that can be taken arbitarily many times over finite inputs which are increasing for the prefix order. For instance, the loop in Example~\ref{exa:emptyN-nonemptyQ-decr} is not iterable: the first letter in the input sets a bound on the number of times it can be taken. The loop in Example~\ref{exa:emptyN-nonemptyQ-incr-bounded} is iterable: it suffices to guess bigger and bigger values of the initial upper bound. However, there can be no infinite run which contains infinitely many occurrences of such loop, as the value that is initially guessed for a given run sets a bound on the number of times the loop can be taken, so it is not $\omega$-iterable.

We show that the notions of iterability and $\omega$-iterability are both characterised by properties on the order between registers of a pair of configurations, which can be summed up into a type, hence opening the way to deciding such properties.

\subsection{\texorpdfstring{$\Q$-types}{Q-types}}%
\label{sec:Q-types}

 In our study, the relative order between registers plays a key role. Such information is summed up by the type of the configuration, interpreted as a configuration in $\Q_+$.

Since we will need to manipulate different types of copies of some set of registers, we adopt the following convention:
\begin{conv}%
  \label{conv:copies}
  In the following, we assume that for a set of registers $R$, $R_1$ and $R_2$ are two disjoint copies of $R$, whose elements are respectively $r_1$ and $r_2$ for $r \in R$. Similarly, $R'$ is a primed copy of $R$, whose elements are $r'$ for $r \in R$. Note that the two can be combined to get $R_1', R_2'$. Note also that primes and indices are also used as usual notations, but no ambiguity should arise.
\end{conv}

\begin{defi}
  For a register valuation $\bar d : R \rightarrow \N$, we define $\tau(\bar d)$ as $\tau_{\Q_+}(\bar d)$ the type of the valuation in $\Q_+$, i.e.\ an FO-formula describing its orbit (such an FO-formula exists by Proposition~\ref{prop:oligoFO} since $\Q_+$ is oligomorphic). Note that such type can be represented e.g.\ by the formula $\bigwedge_{\bowtie \in \{<,>,=\}}\bigwedge_{r, r' \in R \mid \bar d(r) \bowtie \bar d(r')} r \bowtie r' \wedge \bigwedge_{r \in R \mid \bar d(r) = 0} r = 0$. We extend the notion to configurations $(p,\bar d) \in Q \times \N^R$ by letting $\tau((p,\bar d)) = (p, \tau(\bar d))$. Thus, the type specifies the current state, and summarises the information of the order between registers, as well as whether they are equal to $0$ or not.

  We will also need to have access to the relative order between registers of two configurations. Thus, for two register valuations $\bar d_1,\bar d_2 : R \rightarrow \N$, we define $\sigma(\bar d_1,\bar d_2) = \tau_{\Q_+}(\bar d_1 \uplus \bar d_2')$, where $\bar d_1 \uplus \bar d_2'$ is the disjoint union of $\bar d_1$ and of a primed copy of $\bar d_2$, so that the registers of $\bar d_2$ can be distinguished from those of $\bar d_1$. We then have, for all registers $r,s \in R$ and all relations $\bowtie \in \{<,>,=\}$ that $\sigma(\bar d_1,\bar d_2) \Rightarrow r \bowtie s'$ if and only if $\bar d_1(r) \bowtie \bar d_2(s)$. Again, the definition is naturally lifted to configurations by letting $\sigma((p,\bar d_1), (q,\bar d_2)) = (p,q,\sigma(\bar d_1,\bar d_2))$.
\end{defi}
\begin{rem}
Recall that by definition of an orbit, we have that for any register valuations $\bar d_1$
and $\bar d_2$ such that $\tau(\bar d_1) = \tau(\bar d_2)$, there exists an automorphism $\mu \in \aut(\Q_+)$ such that $\mu(\bar d_1) = \bar d_2$.
\end{rem}

The core property is the following:

\begin{pty}
  Let $R$ be a set of registers, and let $\sigma$ be a $\Q$-type defined over $R \uplus R'$, where $R'$ is a primed copy of $R$. We say that $\sigma$ has the property $\star$ for the set of registers $X \subseteq R$ if:
  \begin{itemize}
  \item for all $r \in X$, $\sigma \Rightarrow r \leq r'$
  \item for all $r,s \in X$, if $\sigma \Rightarrow s = s'$ and $\sigma \Rightarrow r \leq s$, then $\sigma \Rightarrow r = r'$
  \end{itemize}
  By extension, for two configurations $C$ and $C'$ over $R$, we say that $C$
  and $C'$ have the property $\star$ for the set of registers $X \subseteq R$ if
  $\sigma(C,C')$ has the $\star$ property for $X$.

  Finally, when $X = R$, we simply state that $\sigma$ has the $\star$ property.
\end{pty}

Such property ensures, for the considered subset of registers, that they cannot induce infinite descending chains nor infinite bounded increasing chains (as both are not feasible in $\N$), if a run loops over configurations whose pairwise type is $\sigma$.

\subsection{\texorpdfstring{Relations between machines over $\N$ and over $\Q_+$}{Relations between machines over N and Q+}}

There is a tight relation between machines operating over $\N$ and over $\Q_+$. First, since $\N$ is a subdomain of $\Q_+$, runs in $\N$ are also runs in $\Q_+$. Over finite runs, by multiplying all data values by the product of their denominators, we can get the converse property.
\begin{prop}%
  \label{prop:Q_N_fin}
  Let $X \subset_f \Q_+$ be a finite subset of $\Q_+$. There exists an
  automorphism $\lambda \in \aut(\Q_+)$ such that $\lambda(X) \subset \N$,
  $\lambda(\N) \subseteq \N$ and $\lambda$ is non-contracting, i.e.\ for all $x,
  y \in \Q_+$, $\size{\lambda(x) - \lambda(y)} \geq \size{x - y}$.
\end{prop}
\begin{proof}
  By writing $X = \left\{\frac{p_1}{q_1}, \dots, \frac{p_n}{q_n}\right\}$, let $K = \prod_i q_i$. Then $\lambda : d \mapsto K d$ is an automorphism satisfying the required properties.
\end{proof}
We then get the following:
\begin{prop}\label{prop:fromQtoN}
  Let $A$ be a $\nra$ over $\Q_+$, and let $\nu$ and $\nu'$ be $\Q$-types.

  If there exist two configurations $C = (p,\bar d), C' = (q,\bar d')$ where $\bar d, \bar d' : R \rightarrow \Q_+$ are such that $\tau(\bar d) = \nu$, $\tau(\bar d') = \nu'$ and if there exists a data word $v \in \Q_+^*$ such that $C \xrightarrow{u} C'$, then there also exist two configurations $D = (p, \bar e)$, $D'= (q, \bar e')$ and a data word $w$ which satisfy the same properties, i.e.\ $\tau(D) = \nu$, $\tau(D') = \nu'$ and $D \xrightarrow{w} D'$; and which belong to $\N$, i.e.\ such that $\bar e,\bar e' : R \rightarrow \N$ and $w \in \N^*$.
\end{prop}
\begin{proof}
  Let $A$ be a $\nra$ over $\Q_+$, and assume that $C \xrightarrow{u} C'$. By applying Proposition~\ref{prop:Q_N_fin} to $X = C(R) \cup C'(R) \cup \mathsf{data(u)}$ (states do not play a role here), we get that $\lambda(C) \xrightarrow{\lambda(u)} \lambda(C')$ is also a run of $A$, and $D = \lambda(C)$, $D' = \lambda(C')$ and $w = \lambda(u)$ satisfy the required properties.
\end{proof}
\begin{rem}
  As a corollary, we obtain that for any $\nra$ $A$ over finite words, $L_{\N}(A) \neq \varnothing$ if and only if $L_{\Q_+}(A) \neq \varnothing$.
\end{rem}
Note that such property does not hold over infinite runs, as witnessed by
Examples~\ref{exa:emptyN-nonemptyQ-decr}
and~\ref{exa:emptyN-nonemptyQ-incr-bounded}.
The property $\star$ ensures that a loop can be iterated
in $\mathbb{N}$, as shown in the next key proposition:
\begin{prop}\label{prop:Qwithstar}
Let $T$ be an \nrt and assume that $B \xrightarrow{u} B'$ following
some sequence of transitions $\pi$,
where $\tau(B)=\tau(B')$, $u\in \mathbb{Q}_+^*$ and
$B$ and $B'$ have the property $\star$.
Then there exists an infinite run $D \xrightarrow{x}$
over the sequence of transitions $\pi^\omega$,
with $x\in \mathbb{N}^\omega$ and $\tau(D) = \tau(B)$.
\end{prop}
Before showing this proposition, let us introduce some intermediate notions:
\begin{defi}
Let $\bar d,\bar d'$ be two register valuations over $\Q_+$
such that $\tau(\bar d)=\tau(\bar d')$. We say that $\bar d'$ is \emph{wider} than $\bar d$
whenever for any $r,s\in R$, we have:
\begin{itemize}
\item $\size{\bar d'(s)- \bar d'(r)} \geq \size{\bar d(s)-\bar d(r)}$
\item $\bar d'(r) \geq \bar d(r)$
\end{itemize}
Note that the second item of the definition is required to ensure that the
interval between $\bar d'(r)$ and $0$ is also wider than the interval between $\bar d(r)$
and $0$.

The notion is extended to configurations by saying that $C' = (p',\bar d')$ is wider than $C = (p, \bar d)$ if $\bar d'$ is wider than $\bar d$. In the following, we only apply this notion to configurations $C$ and $C'$ that have the same state, i.e.\ $p = p'$.
\end{defi}
\begin{prop}%
  \label{prop:starWider}
  Let $\sigma$ be a type over $R \uplus R'$ such that $\sigma_{\mid R} =
  \sigma_{\mid R'}$. If $\sigma$ has the $\star$
  property, then there exist two register valuations $\bar d$ and $\bar d'$ in $\N$ such that
  $\sigma(\bar d,\bar d') = \sigma$ and such that $\bar d'$ is wider than $\bar d$.
\end{prop}
\begin{proof}
  First, assume that there exists some register $r_0 \in R$ such that $\sigma
  \Rightarrow r_0 = 0$ (thus $\sigma \Rightarrow r_0' = 0$, since we assumed that
  $\sigma_{\mid R} = \sigma_{\mid R'}$). If this is not the case, consider
  instead the type $\sigma_0 = \sigma \wedge r_0 = 0 \wedge r'_0 = 0$ over $R
  \uplus \{r_0\}$. Indeed, if $\bar d$ and $\bar d'$ are two
  valuations in $\N$ such that $\sigma(\bar d,\bar d') = \sigma_0$ and $\bar d'$ is wider
  than $\bar d'$, then $\bar d_{\mid R}$ and $\bar d'_{\mid R}$ are two valuations in $\N$
  such that $\sigma(\bar d,\bar d') = \sigma$ and $\bar d'_{\mid R}$ is wider than $\bar d_{\mid R}$.

  Now, for a pair of valuations $(\bar d,\bar d')$, define its shrinking intervals:
  $S(\bar d,\bar d') = \{(r,s) \mid \size{\bar d'(s) - \bar d'(r)} < \size{\bar d(s) - \bar d(r)}\}$, and say
  that $(\bar d,\bar d')$ has $k = \size{S(\bar d,\bar d')}$ shrinking intervals. We need to show
  that given a pair of valuations $\bar e$ and $\bar e'$ such that $\sigma(\bar e,\bar e') =
  \sigma$, if they have $k > 0$ intervals that shrink, we can exhibit a pair of
  valuations $\bar d$ and $\bar d'$ such that $\sigma(\bar d,\bar d') = \sigma$ and which has
  $l < k$ intervals that shrink.

  Thus, let $\bar e$ and $\bar e'$ be two valuations such that $\sigma(\bar e,\bar e') = \sigma$
  and which have $k > 0$ shrinking intervals.
  Then, let $(r,s) \in S(\bar e,\bar e')$; w.l.o.g.\ assume that $\bar e'(s) \geq \bar e'(r)$. As $\bar e$ and
  $\bar e'$ have the same type, we get $\bar e(s) \geq \bar e(r)$. Moreover, as $(r,s) \in S$,
  $\bar e(s) \neq \bar e(r)$, so $\bar e(s) > \bar e(r)$, which implies $\bar e'(s) > \bar e'(r)$. Finally, we
  have that $\bar e'(s) > \bar e(s)$. Indeed, since $\sigma$ has the $\star$ property,
  we have that $\bar e'(s) \geq \bar e(s)$, and
  moreover if we had $\bar e'(s) = \bar e(s)$, we would get that
  $\bar e(r) = \bar e'(r)$ as $\bar e(r) \leq \bar e(s)$, which would mean that $(r,s) \notin
  S(\bar e,\bar e')$.

  Finally, let $M = \max \left(\{\bar e(t) \mid t \in R, \bar e(t) < \bar e'(s)\} \cup \{\bar e'(t)
    \mid t \in R, \bar e'(t) < \bar e'(s)\}  \right)$ be the maximum value seen in $\bar e$ and
  $\bar e'$ which is lower than $\bar e'(s)$. Note that we have $M \geq \bar e(r)$, $M \geq
  \bar e'(r)$ and $M \geq \bar e(s)$.

  Now, let $c = \big(\bar e(s) - \bar e(r)\big) - \big(\bar e'(s) - \bar e'(r)\big)$. As $(r,s) \in
  S$, $c > 0$. Then, consider the automorphism $\mu \in \aut(\Q_+)$ defined as
  $\left\{
    \begin{array}{l}
      x \in [0;M] \mapsto x \\
      x \in ]M;\bar e'(s)] \mapsto M + \frac{\bar e'(s) + c - M}{\bar e'(s) - M}(x - M) \\
      x \in ]\bar e'(s);+\infty[ \mapsto x + c
    \end{array}
  \right.$

It can be checked that for all $x,y \in \Q_+$, $\size{\mu(y) - \mu(x)} \geq \size{y -
  x}$, so $\size{S(\mu(\bar e),\mu(\bar e'))} \leq \size{S(\bar e,\bar e')}$. Now,
we have that $(r,s) \notin S(\mu(\bar e),\mu(\bar e'))$. Indeed,
$\mu(\bar e(s)) = \bar e(s)$, $\mu(\bar e(r)) = \bar e(r)$ and $\mu(\bar e'(r)) = \bar e'(r)$ since
$\bar e(s),\bar e(r),\bar e'(r) \in [0;M]$. Finally, $\mu(\bar e'(s)) = \bar e'(s) + c$.
Overall, we get that $\big(\mu(\bar e(s)) - \mu(\bar e(r))\big) - \big(\mu(\bar e'(s)) -
\mu(\bar e'(r))\big) = \big(\bar e(s) - \bar e(r)\big) - \big(\bar e'(s) + c -
\bar e'(r)\big) = c - c = 0$, which means that $\size{\bar e'(s) - \bar e'(r)} =
\size{\bar e(s) - \bar e(r)}$, so $(r,s) \notin S(\mu(\bar e),\mu(\bar e'))$:
$\size{S(\mu(\bar e),\mu(\bar e'))} < \size{S(\bar e,\bar e')}$.

Since $\mu \in \aut(\Q_+)$, we get that $(\mu(\bar e),\mu(\bar e'))$ is such that
$\sigma(\mu(\bar e),\mu(\bar e')) = \sigma$, so we exhibited a pair of valuations
which has $l < k$ intervals that shrink.

By iteratively applying this process to $(\bar e,\bar e')$ until no shrinking intervals remain, we
get a pair $(\bar d,\bar d')$ which is such that $\sigma(\bar d,\bar d') = \sigma$ and $\bar d'$ is wider
than $\bar d$. Now, by
multiplying all data values in $\bar d$ and $\bar d'$ by the product of their denominators, we
get two valuations $\bar f$ and $\bar f'$ which are in $\N$ such that $\sigma(\bar f,\bar f') =
\sigma$ and $\bar f'$ is wider than $\bar f$.
\end{proof}

We finally need the following technical result:
\begin{prop}%
  \label{prop:widerMorphN}
  Let $\bar d$ and $\bar d'$ be two configurations in $\N$ such that $\bar d'$ is wider than
  $\bar d$. Then, there exists some automorphism $\mu \in \aut(\Q_+)$ such that $\bar d' =
  \mu(\bar d)$ and $\mu(\N) \subseteq \N$.
\end{prop}
\begin{proof}
  Let $\{a_0, \dots, a_k\} = \bar d(R) \cup \{0\}$ and $\{a'_0, \dots, a'_k\} = \bar d'(R)
  \cup \{0\}$, with $0 = a_0 < \dots <
  a_k$ and $a'_0 < \dots < a'_k$. Note that since $\tau(\bar d) = \tau(\bar d')$, we indeed have
  that $\size{\bar d(R) \cup \{0\}} = \size{\bar d'(R) \cup \{0\}}$; moreover, since $\bar d'$
  is wider than $\bar d$, we have that for all $0 \leq i < k$, $a'_{i+1} - a'_i >
  a_{i+1} - a_i$. Consider the following easy lemma (the proof is in Appendix~\ref{app:prooflemnoshrinkbij} for completeness):
\begin{restatable}{lem}{lemnoshrinkbij}\label{lem:noshrinkbij}
  Let $a,b,c,d \in \N$ be such that $a < b$, $c < d$ and $d-c \geq b-a$. Then,
  there exists a function $f: [a;b] \rightarrow [c;d]$ which is increasing and
  bijective, and such that $f([a;b] \cap \N) \subseteq \N$.
\end{restatable}
Then, apply it to each interval $[a_i;a_{i+1}]$ to get a family of increasing
and bijective functions $(\mu_i)_{0 \leq i < k}$ which are such that
$\mu_i([a_i;a_{i+1}]) = [a'_i;a'_{i+1}]$ and $\mu([a_i;a_{i+1}] \cap \N)
\subseteq \N$. Then, let $\mu_k: x \in [a_k;+\infty[ \mapsto x + a'_k - a_k$.
We get that $\mu = \cup_{0 \leq i \leq k} \mu_i \in \aut(\Q_+)$ is such that $\bar d'
= \mu(\bar d)$ and satisfies $\mu(\N) \subseteq \N$.
\end{proof}
We are now ready to prove Proposition~\ref{prop:Qwithstar}:
\begin{proof}[Proof of Proposition~\ref{prop:Qwithstar}]
Let $T$ be an \nrt and assume that $B \xrightarrow[\pi]{u} B'$ following
some sequence of transitions $\pi$,
where $\tau(B)=\tau(B')$, $u\in \mathbb{Q}_+^*$ and
$B$ and $B'$ have the property $\star$.

Let $\sigma = \sigma(B,B')$. By Proposition~\ref{prop:starWider}, we know that
there exist two configurations $C$ and $C'$ in $\N$ such that $\sigma(C,C') = \sigma$
and $C'$ is wider than $C$. Let $\nu \in \aut(\Q_+)$ be some automorphism such that
$\nu(B) = C$ and $\nu(B') = C'$ (such an automorphism exists since $\sigma(B,B') =
\sigma(C,C')$). Then, we have that $C \xrightarrow[\pi]{\nu(u)} C'$. By multiplying
all involved data values by their common denominator (cf
Propostion~\ref{prop:Q_N_fin}), we get two configurations $D$ and $D'$, along
with some data word $v$, all belonging to $\N$, such that $D \xrightarrow[\pi]{v}
D'$, and such that $D'$ is wider than $D$. By
Proposition~\ref{prop:widerMorphN}, there exists an automorphism $\mu \in \aut(\Q_+)$
such that $\mu(D) = D'$ and $\mu(\N) \subseteq \N$. Thus, by letting $x = v
\mu(v) \mu^2(v) \dots$, we get that $D \xrightarrow[\pi]{v} \mu(D)
\xrightarrow[\pi]{\mu(v)} \mu^2(D) \dots$ is a run over $x \in \N^\omega$ over
the sequence of transitions $\pi^\omega$, and $\tau(D) = \tau(B)$.
\end{proof}
A last important property is the following:
\begin{prop}\label{prop:glueN}
  Let $A$ be a $\nra$ over $\Q_+$, and assume that $C \xrightarrow{u} C'$ and that $C'' \xrightarrow{v}$, where $\tau(C') = \tau(C'')$, $u \in \Q_+^*$ and $v \in \N^\omega$. Then there exist $w \in \N^*$, $x \in \N^\omega$ and two configurations $D, D'$ whose valuations take their values in $\N$ such that $D \xrightarrow{w} D' \xrightarrow{x}$, $\tau(C)=\tau(D)$ and $\tau(C')=\tau(D')$.
\end{prop}
\begin{proof}
  Since $\tau(C') = \tau(C'')$, there exists some $\mu \in \aut(\Q_+)$ such that $\mu(C') = C''$, thus $\mu(C) \xrightarrow{\mu(u)} C''$. Now, by applying Proposition~\ref{prop:Q_N_fin} to $X = (\mu(C))(R) \cup C''(R) \cup \mathsf{data}(\mu(u))$, we get that $\lambda(\mu(C)) \xrightarrow{\lambda(\mu(u))} \lambda(C'') \xrightarrow{\lambda(v)}$ satisfies the required property.
\end{proof}

\subsection{Emptiness of automata}

\begin{prop}[Non-emptiness]%
  \label{prop:emptinessN}
  Let $A$ be an \nra over $\N^\omega$. The following are equivalent:
  \begin{enumerate}
    \item $L(A)$ is non-empty
    \item there exist two runs whose input words
    belong to $\mathbb{Q}_+^*$, which are as follows:
    \begin{enumerate}
    \item\label{itm:reachable} $C_0 \xrightarrow{u} C$
    \item\label{itm:star} $D\xrightarrow{v} D'$
   with $\tau(D)=\tau(D')=\tau(C)$, $D$ is a final configuration, and
   $\sigma = \sigma(D,D')$ satisfies $\star$ property.
    \end{enumerate}
  \end{enumerate}
  \end{prop}

\begin{proof}
Assume~(2) holds. The result easily follows from Propositions~\ref{prop:fromQtoN},
~\ref{prop:Qwithstar}
and~\ref{prop:glueN}.

Assume now that $L(A)$ is non-empty.
Let $\rho =
C_0 \xrightarrow{d_0} C_1 \xrightarrow{d_1} C_2 \dots$ be an accepting run over input $x = d_0 d_1 \dots$ in $A$ (where $C_0 =
(q_0,\overline{0})$). For each $i \geq 0$, let $\nu_i = \tau(C_i)$. As $\rho$ is
acccepting and there are only finitely many types, we get that there exists some
type $\nu$ such that the state is accepting and $(q_i,\nu_i) = (q,\nu)$ for infinitely
many $i \in \N$. Let $(C_j)_{j \in \N}$ be an infinite subsequence of
$C_i$ such that for all $j$, $\tau(C_j) = \nu$. Now, colour the
set of unordered pairs as follows:
$c\left(\left\{\tau(C_j),\tau(C_k)\right\}\right) = \sigma(C_j,C_k)$
(where we assume w.l.o.g. that $j < k$).
By Ramsey's theorem, there is an infinite subset such
that all pairs have the same colour $\sigma$. Let $(C_k)_{k \in \N}$ be an infinite
subsequence such that for all $j < k$, $\sigma(C_j,C_k) = \sigma$. Now,
assume that $\sigma$ breaks the $\star$ property. There are two cases:
\begin{itemize}
\item There exists some $r$ such that $\sigma \Rightarrow r > r'$. Then, it means
  that for all $j < k$, $C_j(r) > C_k(r)$. In particular, this means that
  $C_0(r) > C_1(r) > \dots > C_n(r) \dots$, which yields an infinite descending
  chain in $\N$, and leads to a contradiction.
\item There exists some $s$ such that $\sigma \Rightarrow s = s'$ and some $r$ which
  satisfies $\sigma \Rightarrow r < s$ and $\sigma \not \Rightarrow r = r'$. If $\sigma
  \Rightarrow r > r'$, we are back to the first case. Otherwise, it means $\sigma
  \Rightarrow r < r'$. Then, on the one hand, $C_0(r) < C_1(r) < \dots < C_n(r) <
  \dots$. On the other hand, $C_0(s) = C_1(s) = \dots = C_n(s) = \dots$. But we also
  have that for all $k \in \N, C_k(r) < C_k(s) = C_0(s)$. Overall, we get an
  infinite increasing chain which is bounded from above by $C_0(s)$, which
  again leads to a contradiction.
\end{itemize}
Thus, $\sigma$ satisfies the $\star$ property. So, this is in particular the
case for some pair of configurations $C = D = C_k$ and $D' = C_{k'}$ for some $k < k'$ taken
from the last extracted subsequence. Such configurations are such that (recall that the $C_k$ are configurations of an
  accepting run over some input, which is in particular initial):
\begin{enumerate}[(a)]
\item $C_0 \xrightarrow{u} C$
\item $D \xrightarrow{v} D'$.
\end{enumerate}
Moreover, $\tau(D) = \tau(D') = \tau(C)$ and $D$
  is final, by definition of $(C_k)_{k \in \N}$.
\end{proof}
\begin{cor}\label{cor:emptinessN}
  Emptiness for \nra over $\N^\omega$ is decidable in $\PSpace$.
\end{cor}
\begin{proof}
  The algorithm is similar to the one for deciding non-emptiness for \nra over
  oligomorphic domains. Indeed, the sought witness lies in $\Q_+$, which is oligomorphic; it suffices to additionally
  check that the pairwise type of $D$ and $D'$ satisfies the star property.
  Thus, the algorithm initially guesses $\tau(C)$ and $\sigma$. Then, checking
  that there indeed exists a configuration whose type is $\tau(C)$  and that can
  be reached from $C_0$ (item~\ref{itm:reachable} of
  Propostion~\ref{prop:emptinessN}) can be done in the same way as for
  Theorem~\ref{thm:decoligo}, by simulating symbolically (i.e.\ over $\Q$-types)
  a run of the automaton. Now, for item~\ref{itm:star}, the algorithm
  again symbolically simulates a run from $D$, by keeping track of the type of
  the current configuration $\tau(D')$, and additionally of the pairwise type
  $\sigma(D,D')$. Since $\sigma(D,D')$ is a $\Q$-type over $2\size{R}$
  registers, it can be stored in polynomial space; moreover, given a transition
  test $\phi$, it can also be updated in polynomial space.
\end{proof}

\subsection{Functionality}

Following the study of the relationships between $\mathbb{N}$
and $\mathbb{Q}$, we are now ready to provide a characterization
of non functionality over $\mathbb{N}$. Intuitively, it amounts to
finding two pairs of runs whose inputs are in $\mathbb{Q}$: first,
a prefix witnessing a mismatch, and second, an accepting loop
satisfying the $\star$ property to ensure its iterability over
 $\mathbb{N}$.

\begin{prop}[Functionality]
  Let $R\subseteq \N^\omega \times \N^\omega$ be given by an \nrt $T$. The following are equivalent:
  \begin{enumerate}
    \item\label{itm:RnotFun} $R$ is not functional
    \item\label{itm:witnessNotFun} there exist two pairs of runs whose input words
    belong to $\mathbb{Q}_+^*$, which are as follows:
    \begin{enumerate}
    \item\label{itm:funMismatch} $C_0 \xrightarrow{u|u_1} C_1$
    and $C_0 \xrightarrow{u|u_2} C_2$ with $u_1\mismatch u_2$,
    \item\label{itm:funIterable} $D_1\xrightarrow{v} D'_1$
   and $D_2\xrightarrow{v} D'_2$ with $\tau(D_1 \uplus D_2)=\tau(D'_1 \uplus
   D_2)=\tau(C_1 \uplus C_2)$, both runs visit a final state of $T$, and
   $\sigma_i = \tau((D_1 \uplus D_2) \uplus (D'_1 \uplus D'_2))$ satisfies property $\star$.
    \end{enumerate}
  \end{enumerate}
  \end{prop}

\begin{proof}
  First, assume that~(\ref{itm:witnessNotFun}) holds. By applying
  Proposition~\ref{prop:Qwithstar} to the product transducer $T \times T$, there
  exist two configurations $E_1$ and $E_2$ and an infinite data word $x \in
  \N^\omega$ such that $E_1 \xrightarrow{x \mid y_1}$ and $E_2 \xrightarrow{x \mid y_2}$. Moreover, both runs are accepting as
  each finite run is required to visit a final state of $T$.

  Now, since $\tau(E_1 \uplus E_2) = \tau(D_1 \uplus D_2) = \tau(C_1 \uplus
  C_2)$, we can apply Proposition~\ref{prop:glueN} to get two runs $C_0
  \xrightarrow{u' \mid u'_1} F_1 \xrightarrow{x' \mid y_1'}$ and $C_0
  \xrightarrow{u' \mid u'_2} F_2 \xrightarrow{x' \mid y_2'}$. Moreover, since automorphisms
  preserve mismatches, we know that $u'_1 \mismatch u'_2$. Thus, we obtained a
  witness of non-functionality, since we have $(u'x,u'_1y_1'), (u'x,u'_2y_2') \in
  T$, with $u'_1 y_1' \neq u'_2 y_2'$ since $u'_1 \mismatch u'_2$.

  We now assume that $R$ is not functional. By definition, and as we assume $R$
  only contains infinite output words, there are two runs on some input word $x$
  whose outputs mismatch. By splitting both runs after the first mismatch, we
  get two runs $C_0 \xrightarrow{t \mid t_1} B_1 \xrightarrow{w \mid y}$ and
  $C_0 \xrightarrow{t \mid t_2} B_2 \xrightarrow{w \mid z}$ with $t_1 \mismatch
  t_2$ (note that we do not necessarily have that $y = z$). Now, from the
  product transducer $T \times T$, one can define an \nra $A$ with registers $R
  \uplus R'$ recognising the language $L(A) = \left\{w' \mid E_1 \xrightarrow[T]{w' \mid
    y'}, E_2 \xrightarrow[T]{w' \mid z'} \text{ and } \tau(E_1 \uplus E_2) =
  \tau(B_1 \uplus B_2)\right\}$ which starts by guessing some configuration $E_1
  \uplus E_2$ and checks that $\tau(E_1 \uplus E_2) = \tau(B_1 \uplus B_2)$,
  then simulates $T \times T$. Such language is non-empty, since it at least
  contains $w$. By Proposition~\ref{prop:emptinessN}, we get that there exist
  runs whose input words belong to $\Q_+^*$ which are $C_0 \uplus C_0
  \xrightarrow{u'} C_1 \uplus C_2$ and $D_1 \uplus D_2 \xrightarrow{v} D'_1
  \uplus D'_2$ with $\tau(D_1 \uplus D_2) = \tau(D'_1 \uplus D'_2) = \tau(C_1
  \uplus C_2)$, $D_1$ and $D'_1$ final
  and $\sigma(D_1 \uplus D_2,D'_1 \uplus D'_2)$ satisfies the $\star$ property,
  which immediately yields item~\ref{itm:funIterable}.

  Since $C_0 \uplus C_0 \xrightarrow{u'} C_1 \uplus C_2$, by definition of the
  considered $\nra$, we have that $E_1 \xrightarrow[T]{u' \mid u'_1} C_1$ and $E_2
  \xrightarrow[T]{u' \mid u'_2} C_2$, with $\tau(E_1 \uplus E_2)  = \tau(B_1
  \uplus B_2)$. Thus, by applying an automorphism $\mu$ such that $\mu(B_1) = E_1$,
  $\mu(B_2) = E_2$, such runs can be glued with the runs $C_0 \xrightarrow{t
    \mid t_1} B_1$ and $C_0 \xrightarrow{t \mid t_2} B_2$ to yield two runs $C_0
  \xrightarrow{u \mid u_1} C_1$ and $C_0 \xrightarrow{u \mid u_2} C_2$, with $u
  = \mu(t) u'$ and $u_1 = \mu(t_1) u'_1$, $u_2 = \mu(t_2) u'_2$, with $u_1
  \mismatch u_2$ since $\mu(t_1) \mismatch \mu(t_2)$ (recall that automorphisms
  preserve mismatches, since they are bijective), so we get item~\ref{itm:funMismatch}.
\end{proof}
\begin{cor}\label{thm:funN}
  Functionality for relations over $\N^\omega \times \N^\omega$ given by \nrt is decidable in $\PSpace$.
\end{cor}
\begin{proof}
  By Lemma~\ref{lem:smallmis}, if item~\ref{itm:funMismatch} holds, then we can
  assume that the length of $u$ is bounded by $P(\rn(4k), \size{Q}, L)$, where
  $\rn$ denotes the Ryll-Nardzewski function of $\Q_+$, which is exponential.
  Thus, the existence of $u$ can be checked with a counter that is polynomially
  bounded. Then, item~\ref{itm:funIterable} can be checked in polynomial space,
  since it reduces to checking emptiness of the \nra $A$ that we described in
  the above proof.
\end{proof}

\subsection{Next-letter problem}

   We now show that for any function definable by
an \nrt over $\mathbb{N}$, the next-letter problem is
computable.
\begin{lem}\label{lm:nextletterN}
  Let $f: \N^\omega\rightarrow \N^\omega$ be a function defined by an
  \nrt  over $\mathbb{N}$. Then, its next-letter problem is computable.
\end{lem}

\begin{proof}
  The algorithm is in two phases:
  \begin{enumerate}
  \item decide whether there exists a next letter
  \item if there exists a next letter, compute it
  \end{enumerate}

\noindent
  Recall that as input to the next-letter problem, there are
  two finite data words $u,v\in\N^*$ and the goal is to decide whether
  there exists some $d\in\N$ such that for all $uy\in\dom(f)$,
  $v\leq f(uy)$ implies $vd\leq f(uy)$. Let us assume that $f$ is defined by some \nrt $T =
  \tuple{Q,R,\Delta,q_0,\overline{\cd},F}$. Let us explain how to algorithmically realize the two
  latter steps.

  1. To decide the existence of such a $d$, we reduce the problem to a
  functionality problem for an \nrt $T_{uv}$. First, let us define the
  convolution $x_1\otimes x_2$ of two data words $x_1=d_1d_2\dots, x_2=d'_1d'_2\dots \in\N^\omega$ as the
  data word $d_1d'_1d_2d'_2\dots$.   The intuitive idea is to
  construct $T_{uv}$ in such a way that it defines the relation
  $R_{uv} = R_{uv}^1\cup R_{uv}^2$ defined by
  $R_{uv}^i = \{ (ux_1\otimes ux_2, vd^\omega)\mid ux_1,ux_2\in\dom(f),
  vd\preceq f(ux_i)\}$. It is not difficult to see that $R_{uv}$ is a
  function iff the next-letter has a positive answer for $u$ and $v$.
  Once $T_{uv}$ is constructed, checking its functionality is
  possible thanks to Corollary~\ref{thm:funN}.
  Let us now explain how to construct $T_{uv}$. It is the union of two
  transducers $T_{uv}^1$ and $T_{uv}^2$ defining $R_{uv}^1$ and $R_{uv}^2$ respectively. The
  constructions are similar for both: $T_{uv}^1$ ignores the even
  position of the input word after $u$ has been read while $T_{uv}^2$ ignores the odd
  position after $u$ has been read, but they otherwise are defined in the same way. Let us
  explain how to construct $T_{uv}^1$. First, $T_{uv}^1$ makes sure
  that its input $\alpha\otimes \beta$ is such that $u\leq \alpha$ and
  $u\leq \beta$, or equivalently that $\alpha \otimes \beta =
  (u\otimes u)(x\otimes y)$ for some $x,y$. This is possible by
  hardcoding $u$ in the transitions of $T$. Likewise, $T_{uv}^1$ also
  makes sure that $v$ is a prefix of $f(ux)$. It is also possible by
  hardcoding in $T$ the output $v$ (to check for instance that a run
  of $T$ output the $i$th data value $d_i$ of $v$, it suffices when $T$
  outputs a register $r$ on its transitions, to add the test $r =
  d_i$). So, $T_{uv}^1$ simulates a run of $T$ on $ux$ (the odd
  positions of $\alpha\otimes \beta$) and a run of $T$
  on $uy$ (the even positions of $\alpha\otimes \beta$). Once $v$ has
  been consumed by the simulated run on $ux$, the first time this
  simulated run outputs something, the data value is kept in a special
  register and outputted all the time by $T_{uv}^1$.

    2. If we know that there exists a next letter $d$, the only thing
    that remains to be done is to compute it. To do so, we can again
    construct a transducer $T'_{uv}$ which simulates $T$ but makes
    sure to accept only input words that start with $u$, accepted by
    runs which outputs words starting with $v$. This can again be done
    by hardcoding $u$ and $v$ in $T$. Then, to compute a data value $d$, it
    suffices to execute $T'_{uv}$ by computing, at any point $i$, all
    the configurations reached by $T'_{uv}$, and by keeping only those
    which are co-reachable. Testing whether a configuration is
    co-reachable can be done by testing the non-emptiness of an \nra
    starting in this initial configuration. Doing so, the algorithm
    computes all prefixes of accepting runs. Eventually, since there
    exists a next letter $d$, one of the run will output it.
\end{proof}

As a direct corollary of Lemma~\ref{lm:nextletterN}, Theorem~\ref{thm:cont2comp}
and Theorem~\ref{thm:comp2cont}, we obtain:

\begin{thm}\label{thm:compcontN}
    Let $f: \N^\omega\rightarrow \N^\omega$ be a function defined by an
  $\nrt$ over $\mathbb{N}$, and let
  $m:\N\rightarrow \N$ be a total function. Then,
  \begin{enumerate}
  \item $f$ is computable iff $f$ is continuous
  \item $f$ is uniformly computable iff $f$ is uniformly continuous
  \item $f$ is $m$-computable iff $f$ is $m$-continuous
  \end{enumerate}
\end{thm}

%%%%%%%%%%%%%%%%%%%%%%%%%%%%%%%%%%%%%%%%%%%%%%%

\subsection{Uniform continuity}

We now turn to uniform continuity over $\N$, and show that it is enough to consider
uniform continuity over $\Q_+$ and restrict our attention to configurations that are co-reachable
w.r.t. data words over $\N$.

Given a configuration $\Q$-type $\tau$, we say that it is co-reachable in $\N$ if there is an actual configuration $C$ of type $\tau$ which is co-reachable.
\begin{prop}%
  \label{prop:characUnifContN}
Let $T$ be an \nrt over $(\N,\set{<,0})$ and $f= \rel T$.
The following are equivalent:
\begin{itemize}
\item $f$ is uniformly continuous,
\item There is a critical pattern (see Definition~\ref{def:criticalPattern}) with $D_1,D_2$ co-reachable in $\N$.
\end{itemize}
\end{prop}

\begin{proof}
Let $T'$ denote the same transducer as $T$ except that (1) it is over $\Q$ and~(2) it is restricted to configurations which are co-reachable in $\N$.

We claim that $T$ realizes a uniformly continuous function if and only if $T'$ does. From Proposition~\ref{prop:crit}, this is enough to show the expected result.

Any run over $T$ is in particular a run over $T'$, hence if $T$ is not uniformly continuous then, in particular, $T'$ is not either.
Conversely, let us assume that $T'$ is not uniformly continuous. According to Proposition~\ref{prop:crit}, this means that we can exhibit a critical pattern $C_0 \xrightarrow{u|u_1} C_1 \xrightarrow {v|v_1} \mu(C_1)\xrightarrow{w|w_1}D_1$, $C_0 \xrightarrow{u|u_2} C_2 \xrightarrow {v|v_2} \mu(C_2)\xrightarrow{z|w_2}D_2$ such that $D_1,D_2$ are co-reachable. By definition we have that $D_1,D_2$ are co-reachable in $\N$.
Let $i$ denote the mismatching position in the pattern. Let $n>0$, we want to exhibit two inputs that have a common prefix of length at least $n$ but outputs that have a longest common prefix smaller than $i$. We consider the two runs $C_0 \xrightarrow{u|u_1} C_1 \xrightarrow {v|v_1} \mu(C_1)\cdots \mu^n(C_1)\xrightarrow{\mu^n(v)|\mu^n(v_1)}\mu^{n+1}(C_1)\xrightarrow{\mu^n(w)|\mu^n(w_1)}\mu^n(D_1)$ and $C_0 \xrightarrow{u|u_2} C_2 \xrightarrow {v|v_2} \mu(C_2)\cdots \mu^n(C_2)\xrightarrow{\mu^n(v)|\mu^n(v_2)}\mu^{n+1}(C_2)\xrightarrow{\mu^n(w)|\mu^n(w_1)}\mu^n(D_2)$. Note that it could be that $\mu^n$ cancels the mismatch, in which case we can consider $\mu^{n+1}$.
Since $\mu^n(D_1)$ and $\mu^n(D_2)$ are co-reachable, we can use Proposition~\ref{prop:glueN} and we obtain the result.
\end{proof}

As a corollary, we obtain:
\begin{thm}%
  \label{thm:decUnifContN}
  Let $f: \N^\omega\rightarrow \N^\omega$ be a function defined by an
  \nrt  over $\mathbb{N}$. Deciding its uniform continuity is $\PSpace$-c.
\end{thm}
\begin{proof}
We are left to show that uniform continuity of $T'$ is in \PSpace.
This problem is in \PSpace from Theorem~\ref{thm:decoligo}. We have to be careful however, since computing $T'$ explicitly may be too costly. The trick is that checking if a configuration is co-reachable can be done on the fly in \PSpace, using Proposition~\ref{prop:emptinessN} and Corollary~\ref{cor:emptinessN}.

The complexity lower bound can again be obtained by reducing the problem to emptiness of register automata over $(\N, \{=\})$, which is \PSpace-c~\cite{DBLP:journals/tocl/DemriL09}.
\end{proof}

\subsection{Continuity}
We end our study with the property of continuity of \nrt over $\N$.

\begin{asm}
To simplify the following statement, we assume that transducers are equipped
with a distinguished register $r_m$ which along the run stores the maximal data value
seen in the input so far.
\end{asm}
\begin{prop}[Continuity]
  Let $f:\N^\omega \rightarrow \N^\omega$ be given by an \nrt $T$ over
  $\mathbb{N}$. The following are equivalent:
  \begin{enumerate}
    \item\label{itm:NfNotCont} $f$ is not continuous
    \item\label{itm:NfPattern} there exists a critical pattern in
      $\cp(u,v,w,z,C_1,C_2,C_1',C_2',D_1,D_2)$
      with $u,v,w,z \in \mathbb{Q}_+^*$, where $C_1$ is final, $D_1$ and $D_2$
      are co-reachable in $\N$ and such that
      $\sigma = \tau((C_1 \uplus C_2) \uplus (C'_1 \uplus C'_2))$ satisfies the
      $\star$ property for $X = R_1 \cup R_2^l$, where $R_2^l = \{r_2 \in R_2 \mid
      \exists r'_1 \in R'_1, \sigma \Rightarrow r_2 \leq r'_1\}$.
    \end{enumerate}
  \end{prop}
\begin{proof}[Proof of $(\ref{itm:NfNotCont})~\Rightarrow~(\ref{itm:NfPattern})$]
    Assume first that $f$ is not continuous. Let $x \in \N^\omega$, and let
    $(x_n)_{n \in \N}$ be such that $\forall n \in \N, x_n \in \dom(f)$ $x_n
    \xrightarrow[n \infty]{} x$ but $f(x_n) \not \xrightarrow[n \infty]{} f(x)$.
    Up to extracting subsequences, we can assume w.l.o.g. that there exists some
    $k \in \N$ such that for all $n \in \N$,
    % \size{x \wedge x_n} \geq n$ and
    $\size{f(x) \wedge f(x_n)} \leq k$.
    We denote by $\rho$ a run of $T$ over $x$ yielding output $f(x)$, and by
    $(\rho_n)_{n \in \N}$ a sequence of runs of $T$ such that $\rho_n$ yields
    output $f(x_n)$ over input $f(x)$.

    First, since the set $\Delta$ of transitions of $T$ is finite,
    $\Delta^\omega$ is compact (cf Remark~\ref{rem:infiniteNotCompact}), so
    $(\rho_n)_{n \in \N}$ admits a converging subsequence, so we can assume
    w.l.o.g. that $(x_n)_{n \in \N}$ is such that $(\rho_n)_{n \in \N}$
    converges to some infinite sequence of transitions $\pi$.
    \begin{rem}
      Note however that $\pi$ is not necessarily a run over some concrete data
      word: since transducers are allowed to guess arbitrary data value, we do not
      have the property that after reading input $u$, the configuration of the
      transducer takes its values in $\mathsf{data}(u) \cup \{0\}$, which would
      allow to build a concrete run by extracting subsequences of runs whose
      configurations coincide on longer and longer prefixes. If we disallow
      guessing, such property is restored, which greatly simplifies the proof.
      Indeed, then, we can require that $\sigma$ has the $\star$ property,
      without allowing some wild registers to go astray (namely those that are
      above $r'_k$, which necessarily correspond to data values that have been guessed,
      otherwise their content is at most equal to the one of $r_m$), and having
      two runs that can be instantiated by actual data words allow to extract
      $\sigma$ with a Ramsey argument which is similar to the one used for
      emptiness (see Proposition~\ref{prop:emptinessN}).
    \end{rem}
    Now, let $(E_i)_{i \in \N}$ be the sequence of configurations of $\rho$,
    and, for each $\rho_n$, let $(C_{n,i})_{i \in \N}$ be its corresponding
    sequence of configurations. Then, for all $0 \leq i < j$, let $\tau_{n,i} =
    \tau(C_{n,i})$ and $\sigma_{n,i,j} = \sigma(E_i \uplus C_{n,i}, E_j \uplus
    C_{n,j})$. Since the $\tau_{n,i}$ and $\sigma_{n,i,j}$ take their values in
    a finite set, by using compactness arguments, we can extract two
    subsequences $(\tau'_{n,i})_{n,i \in \N}$ and $(\sigma'_{n,i,j})_{n \in \N,
      0 \leq i < j}$ which respectively converge to $(\tau'_i)_{i \in \N}$ and
    to a family $(\sigma_{i,j})_{0 \leq i < j}$. Formally, we require that for
    all $M \geq 0$, there exists some $N \geq 0$ such that for all $n \geq N$,
    we have that $\tau'_{n,i} = \tau'_i$ and $\sigma'_{n,i,j} = \sigma_{i,j}$
    for all $0 \leq i < j \leq M$ (note that for types $\tau'_{n,i}$, this
    corresponds to convergence for infinite words). By first restricting to
    final configurations in $\rho$ (we know there are infinitely many since
    $\rho$ is accepting), we can assume that all $E_i$ are final. We further
    restrict to $E_i$ which all have the same type $\nu$ (there is at least one
    type which repeats infinitely often). To avoid
    cluttering the notations, we name again $(\rho_n)_{n \in \N}$ the
    corresponding subsequence of runs, $(\tau_{n,i})_{n,i \in \N}$ the
    corresponding types, and $(\sigma_{n,i,j})_{n,i,j}$ their associated family
    of pairwise types. Finally, by applying Ramsey's theorem to $(\tau_i)_{i \in
      \N}$ and $(\sigma_{i,j})_{0 \leq i < j}$, we know that there exists some
    type $\tau$, some pairwise type $\sigma$ and some infinite subset $I
    \subseteq \N$ such that for all $i,j \in I$ such that $i < j$, $\tau_i =
    \tau$ and $\sigma_{i,j} = \sigma$. For simplicity, we reindex the $(E_j)_{j
      \in I}$ and the $(C_{n,j})_{j \in I}$ over $\N$, that we again name
    $(E_j)_{j \in \N}$ and $(C_{n,i})_{i \in \N}$.

    Now, assume by contradiction that $\sigma$ does not have the property
    $\star$ for $X = R_1 \cup R_2^l$. There are two cases, that we treat
    iteratively:
    \begin{itemize}
    \item There exists some $r \in X$ such that $\sigma \Rightarrow r > r'$.
      There are two subcases:
      \begin{itemize}
      \item If there exists such $r \in R_1$, then we get that for all $0 \leq i
        < j$, $E_i(r) > E_j(r)$, which immediately yields an infinite descending
        chain in $\N$, and hence a contradiction.
      \item If there exists such $r \in R_2^l$, then let $r'_M$ be such that
        $\sigma \Rightarrow r \leq r'_M$. Since $r_M \in R_1$, we know that
        $\sigma \Rightarrow r_M \leq r'_M$, otherwise we are back to the
        previous case. Then, let $N \in \N$ be such that $\sigma_{N,i,j} =
        \sigma$ for all $0 \leq i < j \leq B = E_1(r'_M) + 1$. Such $N$
        necessarily exists since we took the $(\rho_n)_{n \in \N}$ such that the
        $(\sigma_{n,i,j})_{0 \leq i < j}$ converges. Thus, $E_1(r'_M) \geq
        C_{N,0}(r) > C_{N,1}(r) > \dots C_{N,B}(r)$, which again yields a
        contradiction.
      \end{itemize}
    \item Assume now that there exists $r,s \in X$ such that $\sigma \Rightarrow
      s = s'$, $\sigma \Rightarrow r \leq s$ but $\sigma \Rightarrow r < r'$.
      There are again two subcases:
      \begin{itemize}
      \item If $s \in R^1$, then we know that for all $i \in \N$, $E_i(s) =
        E_0(s)$. Now, if $r \in R^1$, then we get that $E_0(r) < E_1 (r) <
        \dots$ but for all $i \in \N$, $E_i(r) < E_i(s)$, which leads to a
        contradiction. If $r \in R_2^l$, then let $N$ be such that $\sigma_{N,i,j}
        = \sigma$ for all $0 \leq i < j \leq B = E_0(s) + 1$. We then get a
        strictly increasing chain $C_{N,0}(r) < C_{N,1}(r) < \dots < C_{N,B}(r)$
        of length $B$ which is bounded from above by $B - 1$, which does not
        exist in $\N$.
      \item If $s \in R_2^l$ then let $r'_M \in R_1$ be such that $\sigma
        \Rightarrow s \leq r'_M$. Then, let $B = E_1(r'_M) + 1$, and let $N$ be
        such that $\sigma_{N,i,j} = \sigma$ for all $0 \leq i < j \leq B$. If $r
        \in R^1$ then we have that $E_0(r) < E_1(d) < \dots < E_B(r) \leq E_B(s)
        = E_0(s) \leq E_1(r'_M)$; similarly if $r \in R_2^l$ we get $C_{N,0}(r) <
        C_{N,1}(r) < \dots < C_{N,B}(r) \leq C_{N,B}(s) = C_{N,0}(s) \leq
        E_1(r_M')$. In both cases, this leads to a contradiction.
      \end{itemize}
    \end{itemize}
    As a consequence, $\sigma$ indeed has the property $\star$ for $X = R_1 \cup
    R_2^l$.

    Now, it remains to exhibit a critical pattern from $\rho$ and the (last)
    extracted subsequence $(\rho_n)_{n \in \N}$ and the corresponding inputs
    $(x_n)_{n \in \N}$. Let $k \in \N$ be such that for all $n \in \N$,
    $\size{f(x) \wedge f(x_n)} \leq k$. On the one hand, we have $\rho = E_0
    \xrightarrow{u_0 \mid v_0} E_1 \xrightarrow{u_1 \mid v_1} E_2 \dots$, where
    $x = u_0 u_1 \dots$ and such
    that all $(E_i)_{i > 0}$ are final and they all have the same type. On the
    other hand, each $\rho_n$ can be written $\rho_n = C_{n,0}
    \xrightarrow{u^n_0 \mid v^n_0} C_{n,1} \xrightarrow{u^n_1 \mid v^n_1}
    C_{n,2} \dots$, where for all $j \geq 0, \size{u^n_j} = \size{u_j}$ (the
    configurations $C_{n,i}$ are synchronised with those of $E_i$, as we always
    extracted subsequences in a synchronous way) and $u^n_0 u^n_1 \dots = x_n$.
    Now, let $M$ be such that $\forall n \geq M$, we have $\forall 0 \leq i < j
    \leq B + 3, \sigma_{n,i,j} = \sigma$, where $B = 2k$.
    Finally,
    take some $l \geq M$ such that $\size{x_l \wedge x} \geq \size{u_0 u_1 \dots
      u_{B+1}}$. There are two cases:
    \begin{itemize}
    \item $v_0 \dots v_B \mismatch v^l_0 \dots v^l_B$. Then, we exhibited a
      critical parttern with two runs $E_0 \xrightarrow{u_0 \dots u_B \mid v_0
        \dots v_B} E_{B+1} \xrightarrow{u_{B+1} \mid v_{B+1}} E_{B+2}
      \xrightarrow{u_{B+2} \mid v_{B+2}} E_{B+3}$ along with $C_{l,0} \xrightarrow{u^l_0
        \dots u^l_B \mid v^l_0 \dots v^l_B} C_{l,B+1} \xrightarrow{u^l_{B+1}
        \mid v^l_{B+1}} C_{l,B+2} \xrightarrow{u^l_{B+2} \mid v^l_{B+2}}
      C_{l,B+3}$. First, the outputs indeed mismatch, by hypothesis, so we are
      in case~\ref{itm:critMismatch} of the definition of a critical pattern (see Definition~\ref{def:criticalPattern}).
    \item $v_0 \dots v_B \match v^l_0 \dots v^l_B$. Then, since $v_0 \dots v_B
      \leq f(x)$ and $v^l_0 \dots v^l_B \leq f(x_l)$, and since $\size{f(x)
        \wedge f(x_l)} \leq k$, we get that $\size{v_0 \dots v_B} \leq k$ or
      $\size{v^l_0 \dots v^l_B} \leq k$. Now, there are again two cases:
      \begin{itemize}
      \item There exists some $i \leq B$ such that $v_i = v^l_i = \varepsilon$. Then, we exhibited a critical pattern with two runs $E_0 \xrightarrow{u_0 \dots u_{i-1} \mid v_0 \dots v_{i-1}} E_i \xrightarrow{u_i \mid \varepsilon} E_{i+1} \xrightarrow{u_{i+1} \dots u_m \mid v_{i+1} \dots v_m} E_{m+1}$ and $C_{l,0} \xrightarrow{u^l_0 \dots u^l_{i-1} \mid v^l_0 \dots v^l_{i-1}} C_{l,i} \xrightarrow{u^l_i \mid \varepsilon} C_{l,i+1} \xrightarrow{u^l_{i+1} \dots u^l_m \mid v^l_{i+1} \dots v^l_m} C_{l,m+1}$, where $m$ is such that $v^l_0 \dots v^l_m \mismatch v_0 \dots v_m$ (such $m$ exists since $f(x_n) \mismatch f(x)$). We are then in case~\ref{itm:critTwoEmpty} of the definition of a critical pattern (see Definition~\ref{def:criticalPattern}).
      \item Otherwise, there necessarily exists some $i \leq B$ such that $v_i = \varepsilon$ or $v^l_i = \varepsilon$ since $\size{v_0 \dots v_B} \leq k$ or $\size{v^l_0 \dots v^l_B} \leq k$. We assume that $v_i = \varepsilon$, the reasoning is symmetric if instead $v^l_i = \varepsilon$. Necessarily, $v^l_0 \dots v^l_{i-1} \mismatch f(x)$, otherwise we can find some $i$ such that both $v_i = v^l_i$ and we are back to the previous case. Then, we exhibited a critical pattern with two runs $E_0 \xrightarrow{u_0 \dots u_{i-1} \mid v_0 \dots v_{i-1}} E_i \xrightarrow{u_i \mid \varepsilon} E_{i+1} \xrightarrow{u_{i+1} \dots u_m \mid v_{i+1} \dots v_m} E_{m+1}$ and $C_{l,0} \xrightarrow{u^l_0 \dots u^l_{i-1} \mid v^l_0 \dots v^l_{i-1}} C_{l,i} \xrightarrow{u^l_i \mid v^l_i} C_{l,i+1} \xrightarrow{u^l_{i+1} \dots u^l_m \mid v^l_{i+1} \dots v^l_m} C_{l,m+1}$, where $m$ is such that $v^l_0 \dots v^l_m \mismatch v_0 \dots v_m$ (such $m$ exists since we assumed that $v^l_0 \dots v^l_{i-1} \mismatch f(x)$). We are then in case~\ref{itm:critOneEmpty} of the definition of a critical pattern (see Definition~\ref{def:criticalPattern}).
      \end{itemize}
    \end{itemize}
    Finally, in all the considered cases, the last configuration $D_1 = E_j$ of the first run of the critical pattern (the one which is a prefix of $\rho$) is final, since we extracted the $(E_i)$ so that we only kept the final ones. Also, both $D_1 = E_j$ and $D_2 = C_{l,j}$ are co-reachable in $\N$ since they are configurations of accepting runs in $\N$, which concludes the proof of $(\ref{itm:NfNotCont})~\Rightarrow~(\ref{itm:NfPattern})$.
  \end{proof}
\begin{proof}[Proof of $(\ref{itm:NfPattern})~\Rightarrow~(\ref{itm:NfNotCont})$]

We assume now that $(\ref{itm:NfPattern})$ holds and show that $f$ is not
continuous. Intuitively, following notations of the critical pattern, we will
prove that $f$ is not continuous at some input $x=u'v'\mu(v')\mu(v')\ldots $,
where $u',v'$ elements of $\N^+$, images of $u,v$ by some automorphism.

Thus, let us consider two runs $C_0
\xrightarrow{u|u_1}C_1\xrightarrow{v|v_1}\mu(C_1)$ with $C_1$ final, and $C_0
\xrightarrow{u|u_2}C_2\xrightarrow{v|v_2}\mu(C_2) \xrightarrow{z|w_2}D_2$ with
$D_2$ co-reachable in $\N$. Two cases can occur: either $u_1\mismatch
u_2$, or $v_2=\epsilon$ and $u_1\mismatch u_2w_2$. As in the proof in the
oligomorphic setting, we want to show that it is possible to build from the
first run an infinite run looping forever along (some renaming of) $v$, and from
the second run a family of infinite runs, looping more and more. While this is
directly true in some oligomorphic data domain, as we saw before, iteration in
$\N$ is more tricky.

These runs can be seen as finite runs in the transducer $T\times T$, with twice
as many registers as $T$, which we denote by $R_1$ for the first copy, and $R_2$
for the second. By assumption, $\sigma = \tau((C_1 \uplus C_2) \uplus (C'_1
\uplus C'_2))$ satisfies the $\star$ property for $X = R_1 \cup R_2^l$. If we had
that $\sigma$ satisfies the $\star$ property for $R_1\cup R_2$, then we could
immediately deduce by Proposition~\ref{prop:Qwithstar} that the two loops on $v$
could be iterated infinitely many times. However, the weaker hypothesis we have
will only allow us to show that the loop from $C_1$ can indeed by
$\omega$-iterated, while the loop from $C_2$ can be iterated as many times as we
want, but finitely many times. To this end, we have to take care of registers in
$R_2\setminus R_2^l$ to show that these runs indeed exist.

As we assumed that there is a register $r_m$ storing the maximal data value appearing
in the input word, the definition of the set $R_2^l$ ensures that every register
not in $R_2^l$ has its value coming from a guess along the run $C_0
\xrightarrow{u}C_2$. This is directly linked with the technical difficulty
inherent to the presence of guessing. In particular, this allows us to consider
different runs, in which the input data word is not modified, but the guessed
values are.

As $C_2$ and $\mu(C_2)$ have the same type, and registers not in $R_2^l$ have
``large'' values, their behaviour along the cycle (w.r.t. types) $C_2
\xrightarrow{v}\mu(C_2)$ is very restricted: they can only increase or decrease
at each step. In particular, if one starts from some configuration $C_2$ in
which values of registers not in $R_2^l$ are very far apart, this will allow to
iterate the cycle several times. More precisely, for any integer $n$, we can
compute an integer $n'$ and values for the guesses of registers not in $R_2^l$
which are pairwise $n'$ far apart, ensuring that the cycle can be repeated at
least $n$ times.

Thus, the previous reasoning allows us to replicate the proof of
Proposition~\ref{prop:Qwithstar} to show that there exist a word $v'\in\N^*$ and
an automorphism $\alpha$ preserving $\N$, such that:
\begin{itemize}
\item $C'_1 \xrightarrow{x|y}$, with $\tau(C'_1)=\tau(C_1)$ and $x =
  v'\alpha(v')\alpha^2(v')\ldots$
\item there exists $C'_2$ and, for all $n$, some $C'_2(n)$ with $C'_2(n)_{|R_2^l}
  = C'_{2|R_2^l}$, $\tau(C'_2(n)) = \tau(C_2)$, and $C'_2(n) \xrightarrow{v'}
  \alpha(C'_2(n)) \ldots \xrightarrow{\alpha^n(v')} \alpha^{n+1}(C'_2(n))$
\end{itemize}

\noindent
Using previous remark on guessed values, together with
Proposition~\ref{prop:glueN} applied on $C_0\xrightarrow{u}C_1$ and
$C_0\xrightarrow{u}C_2$, and Proposition~\ref{prop:fromQtoN} applied on
$\mu(C_2)\xrightarrow{z|w_2} D_2\xrightarrow{z_l}$, we end up with:
\begin{itemize}
\item a run $C_0 \xrightarrow{u'|u'_1}E_1\xrightarrow{x'|y'} $, with $x' =
  v''\alpha(v'')\alpha^2(v'')\ldots$
\item for every $n$, there is a run $C_0 \xrightarrow{u'|u'_2}E_2
  \xrightarrow{v''} \alpha(E_2) \ldots \xrightarrow{\alpha^n(v'')}
  \alpha^{n+1}(E_2)\xrightarrow{z'|w'_2} D'_2 \xrightarrow{z'_l} $
\item $u'_1 \mismatch u'_2$ or $u'_1 \mismatch u'_2w'_2$
\end{itemize}
Hence, this proves the non functionality as we have a sequence of runs whose
inputs converge towards $u'x'$ but whose outputs mismatch.
\end{proof}

\begin{thm}\label{thm:contN}
Continuity of \nrt on $\N$ is $\PSpace$-c.
\end{thm}

\begin{proof}
Small critical pattern property (Claim~\ref{clm:smallmis}) on oligomorphic data
can easily be adapted to $\N$. Indeed, the only difference lies in the
loops on $v$, but is is important to notice that loop removal
used to reduce length of runs preserves the extremal configurations
(see Proposition~\ref{prop:loop}), hence it preserves the type of the
global run, here
the type $\sigma= \tau((C_1 \uplus C_2) \uplus (C'_1 \uplus C'_2))$.

Then, one can proceed similarly by guessing such a small critical
pattern, and checking mismatches with counters.
Note that we have to verify in addition that $D_2$ is co-reachable
in $\N$, but again, as already detailed in previous proofs, this property
can be decided on-the-fly.

Finally, the \PSpace lower bound can again be obtained by a reduction to the emptiness problem of register automata over $(\N, \{=\})$, which is \PSpace-c~\cite{DBLP:journals/tocl/DemriL09}.
\end{proof}

\subsection{Transfer result}
We have extended our result to one non-oligomorphic structure, namely $(\N,\set{<,0})$, which is a substructure of $(\Q,\set{<,0})$. If we want to extend the result to other substructures of $(\Q,\set{<})$, say $(\Z,\set{<})$, we could do the same work again and study the properties of iterable loops in $\Z$. However, a simpler way is to observe that $(\Z,\set{<})$ can be simulated by $\N$, using two copies of the structure.
We show a quite simple yet powerful result: given a structure $\D$, if the structure $\D'$ can be defined as a quantifier-free interpretation over $\D$ then the problems of emptiness, functionality, continuity, \etc reduce to the same problems over $\D'$.

A \emph{quantifier-free interpretation} of dimension $l$ with signature $\Gamma$ over a structure $(\D,\Sigma)$ is given by $\qf[\Sigma]$ quantifier-free formulas: a formula $\phi_{\text{domain}}(x_1,\ldots,x_l)$, for each constant symbol $\mathsf c$ of $\Gamma$ a formula $\phi_{\mathsf c}(x_1,\ldots,x_l)$ and for each relation symbol $R$ of $\Gamma$ of arity $r$ a formula $\phi_{R}(x_1^1,\ldots,x_l^1,\ldots,x_1^r,\ldots,x_l^r)$.
The structure $\D'$ is defined by the following:
a domain $D'=\set{(d_1,\ldots,d_l)|\ \D\models \phi_{\text{domain}}(d_1,\ldots,d_l)}$; an interpretation $(d_1,\ldots,d_l)\in D'$ for each constant symbol $\mathsf c$ such that $\D\models \phi_{\mathsf c}(d_1,\ldots,d_l)$ (we assume that there is a unique possible tuple satisfying the formula, which can be syntactically ensured); an interpretation $R^\D=\set{(x_1^1,\ldots,x_l^1,\ldots,x_1^r,\ldots,x_l^r)|\ \D\models \phi_{R}(x_1^1,\ldots,x_l^1,\ldots,x_1^r,\ldots,x_l^r)}$ for each relation symbol $R$.

\begin{thm}%
\label{thm:transfer}
Let $\D'$ be a quantifier-free interpretation over $\D$. Let $P$ denote a decision problem among non-emptiness, functionality, continuity, Cauchy continuity or uniform continuity.
There is a \PSpace reduction from $P$ over $\D'$ to $P$ over $\D$.
\end{thm}
\begin{proof}
Let $R\subseteq \D'^\omega\times \D'^\omega$ be given by an \nrt $T$.
If we assume that $\D'$ is an $l$-dimension interpretation of $\D$, then we can view $R$ as a relation $P\subseteq (\D^l)^\omega\times (\D^l)^\omega$. Note that $P$ is empty (resp.~functional, continuous, Cauchy continuous, uniformly continuous) if and only if $R$ is.

Moreover, since $\D'$ is an interpretation, one can construct an \nrt $S$ which realizes $P$. It uses $l$ registers for every register of $T$ plus $l-1$ registers to store the input. In its first $l-1$ transitions it just stores the input and every $l$ transition, it simulates a transition of $T$, just by substituting the formulas of the interpretation for the predicates. As usual, we do not construct $S$ explicitly, but we are able to simulate it using only polynomial space.
\end{proof}

As a direct corollary we can, in particular, transfer our result over $(\N,\{<,0\})$ to $(\Z,\{<,0\})$:
\begin{cor}
The problems of non-emptiness, functionality, continuity, Cauchy continuity, uniform continuity over $(\Z,\{<,0\})$ are in \PSpace.
\end{cor}

\begin{proof}
From Theorem~\ref{thm:transfer}, using the fact that $(\Z,\{<,0\})$ is given by the following two-dimensional interpretation over $(\N,\{<,0\})$:
$\phi_{\text{domain}}:=(x=0)\vee (y=0)$, $\phi_0:=(x=y=0), \phi_<:=(x_1=x_2=0\wedge(y_1<y_2))\vee(y_1=y_2=0\wedge(x_1>x_2))\vee (x_2=y_1=0\wedge\neg(x_1=y_2=0))$.
\end{proof}

\begin{rem}
Of course our transfer result applies to many other substructures of $(\Q,\{<,0\})$, such as the ordinals $\omega +\omega$, $\omega\times \omega$, \etc.
\end{rem}

\begin{rem}
  Considering first-order interpretation (i.e.\ where $\phi_{\text{domain}}$, etc are first-order formulas instead of quantifier-free ones) would yield too much expressive power. Indeed, $(\N,\set{+1,0})$, where $+1(x,y)$ holds whenever $y = x+1$, can easily be defined as a (one-dimensional) first-order interpretation of $(\N,\set{<,0})$, by letting $\phi_{+1}(x,y) = y > x \wedge \neg \exists z, x < z < y$. However, over such a domain, register automata coincide with counter machines, which is a Turing-complete model (with two or more counters); in particular, emptiness of such register automata is easily shown undecidable (with two or more registers).
\end{rem}

\section*{Future work}
We have given tight complexity results in a family of oligomorphic structures, namely polynomially decidable ones. An example of an exponentially decidable oligomorphic structure is the one of finite bitvectors with bitwise xor operation. Representing the type of $k$ elements may require exponential sized formulas (for example stating that a family of $k$ elements is free, \ie any non-trivial sum is non-zero). The same kind of proof would give \ExpSpace algorithms over this particular data set (for the transducer problems we have been studying). One could try to classify the different structures based on the complexity of solving these kinds of problems.

We have been able to show decidability of several transducer problems over the data set $\N$ with the linear order. This was done using two properties: (1) that $\N$ is a substructure of $\Q$ and~(2) that we were able to characterize the iterable loops in $\N$. Moreover we can transfer the result to other substructures of $\Q$, \eg $\Z$.
One possible extension would be to investigate data sets which have a tree-like structure, \eg the infinite binary tree. There exists a tree-like oligomorphic structure, of which the binary tree is a substructure. Studying the iterable loops of the binary tree may yield a similar result as in the linear order case. But, as it turns out, the notion of quantifier-free interpretation can directly yield decidability, as the binary tree is a quantifier-free interpretation of $\N$, thanks to a result of Demri and Deters~\cite[Section 3]{DBLP:journals/logcom/DemriD16}. Indeed, they show that words over a finite alphabet of size $k$ with the prefix relation (which corresponds to a $k$-ary tree) can be encoded in $\N$. However, the question remains open for non-tree-like structures that are substructures of an oligomorphic structure, for which we might be able to characterize iterable loops.

There are several classical ways of extending synthesis results, for
instance considering larger classes of specifications, larger classes
of implementations or both. In particular, an interesting direction is
to consider non-functional specifications. As mentioned in
Introduction however, and as a motivation for studying the functional
case,
enlarging both (non-functional) specifications and implementations to an asynchronous setting leads to
undecidability. Indeed, already in the finite alphabet setting, the
synthesis problem of deterministic transducers over $\omega$-words
from specifications given by non-deterministic
transducers is undecidable~\cite{CarayolL14}. A simple adaptation of
the proof of~\cite{CarayolL14} allows to show that in this finite
alphabet setting, enlarging the class of implementations to any
computable function also yields an undecidable synthesis problem. An
interesting case however is yet unexplored already in the finite
alphabet setting: given a synchronous specifications, as an
$\omega$-automaton, is to possible to synthesise a computable function
realizing it? In~\cite{DBLP:journals/corr/abs-1209-0800,DBLP:conf/csl/FridmanLZ11},
this question has been shown to be decidable for specifications with
\emph{total} input domain (any input word has at least one correct
output by the specification). More precisely, it is shown that
realizability by a continuous function is decidable, but it turns out
that the synthesised function is definable by a
deterministic transducer (hence computable). When the domain of the
specification is partial, the situation changes drastically:
deterministic transducers may not suffice to realize a
specification realizable by a computable function. This can be seen by considering the (partial) function $g$ of
Introduction, seen as a specification and casted to a finite alphabet
$\{a,b,c\}$: it is not computable by a deterministic transducer, since
it requires an unbounded amount of memory to compute this function.

\bibliographystyle{alphaurl}
\bibliography{Bibliography}

\appendix

\section{Proof of Lemma~\ref{lem:bezout}}%
\label{sec:bezout}
In order to show the lemma we use the result from~\cite[Theorem 9]{MajewskiH94}:
\begin{thm}%
  \label{thm:bezout}
  Let $p_1,\ldots,p_k\in \Z{\setminus}{\set{0}}$ be non-zero integers.
  Then there exist integers $z_1,\ldots, z_k \leq \max(|p_1|,\ldots,|p_k|)$ such that:
  \[z_1p_1+\ldots +z_{k}p_k=\gcd(p_1,\ldots,p_k)\]
\end{thm}
It is more convenient to split the integers into two groups depending on their signs.
Let $p_0,\ldots,p_k$ be positive integers and $q_0,\ldots, q_l$ be negative integers, let $d=\gcd(p_0,\ldots,p_k,q_0,\ldots,q_l)$ and let $M=\max(p_0,\ldots,p_k,-q_0,\ldots, -q_l)$. Using the previous theorem, we know there are integers $m_0,\ldots, m_k$ and $n_0,\ldots,n_l$ all in $[-M,M]$ so that $\sum_{0\leq i\leq k} m_{i}p_i+\sum_{0\leq j\leq l} n_{i}q_i=d$. We modify the coefficients in the following way:
$m_0'=m_0-M\sum_{1\leq j\leq l}q_j$, $m_i'=m_i-Mq_0$ for $i\in \set{1,\ldots,k}$, $n_0'=n_0+M\sum_{1\leq i\leq k}p_i$, $n_j'=n_j+Mp_0$ for $j\in \set{1,\ldots,l}$. Note that all the new coefficients are positive. We only have left to check that the sum is unchanged.

$\begin{array}{l}\sum_{0\leq i\leq k} m_i'p_i+\sum_{0\leq j\leq l} n_i'q_i=\\
  \sum_{0\leq i\leq k} m_{i}p_i-p_0M\sum_{1\leq j\leq l}q_j-\sum_{0\leq i\leq k}Mq_0p_i\\+\sum_{0\leq j\leq l} n_{i}q_i+q_0M\sum_{1\leq i\leq k}p_i+\sum_{0\leq j\leq l}Mp_0q_j\\
  =\sum_{0\leq i\leq k} m_{i}p_i+\sum_{0\leq j\leq l} n_{i}q_i
\end{array}$

We have managed to obtain positive coefficients, all smaller than $M^3$.

\section{Proof of Lemma~\ref{lem:noshrinkbij}}%
\label{app:prooflemnoshrinkbij}
\lemnoshrinkbij*
\begin{proof}
 There are two cases:
 \begin{itemize}
 \item If $d - c = b - a$, then take $f: x \rightarrow x + c - a$.
 \item Otherwise, define $f$ as: for all $x \in [a;b-1]$ (note that such
   interval can be empty), let $f(x) = x + c - a$ as above. Then, for $x \in
   [b-1;b]$, let $e = c+(b-1-a)$ and $f(x) = d - (b-x)(d-e)$. We have $f(b-1) =
   e$, which is consistent with the definition of $f$ on $[a;b-1]$, and $f(b) =
   d$, and $f$ is moreover increasing and bijective. Finally, $f(b-1) \in \N$
   and $f(b) \in \N$. Overall, $f$ satisfies the required properties.
   \begin{figure}[ht]
     \begin{tikzpicture}[xscale=1.2,yscale=0.8]
       \node[inner sep=0pt,label={left:$a$}] (a)  at (0,0) {\textbullet};
       \node[inner sep=0pt,label={right:$c$}] (c) at (2,0.5) {\textbullet};
       \draw[->,>=stealth] (a) -- (c);
       \node[rotate=11] at (1,1.75) {$\cdots$};
       \node[inner sep=0pt,label={left:$b-1$}] (bm1)  at (0,3) {\textbullet};
       \node[inner sep=0pt,label={right:$e=c+(b-1-a)$}] (e) at (2,3.5) {\textbullet};
       \draw[->,>=stealth] (bm1) -- (e);
       \begin{scope}[opacity=0.5]
       \node[inner sep=0pt,label={left:$b-1/2$}] (bm12)  at (0,3.5) {\textbullet};
       \node[inner sep=0pt,label={right:$e+\dfrac{d-e}{2}$}] (de2) at (2,4.75) {\textbullet};
       \draw[->,>=stealth] (bm12) -- (de2);
       \end{scope}
       \node[inner sep=0pt,label={left:$b$}] (b)  at (0,4) {\textbullet};
       \node[inner sep=0pt,label={right:$d$}] (d) at (2,6) {\textbullet};
       \draw[->,>=stealth] (b) -- (d);
\end{tikzpicture}
\end{figure} \qedhere
 \end{itemize}
\end{proof}

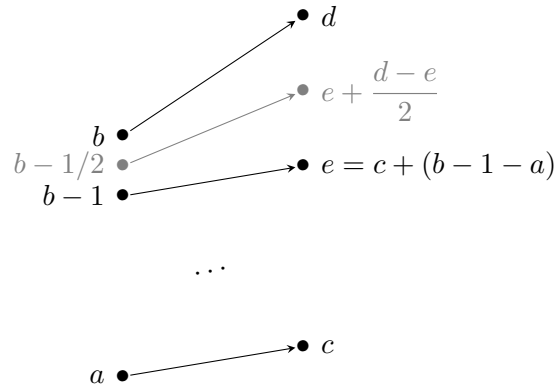
\end{document}